\documentclass[a4paper,11pt,openright]{article} 
\usepackage[english]{babel} 
\usepackage[latin1]{inputenc} 
\usepackage{amsmath,ascii,newunicodechar,amssymb,amsthm,amsfonts,amscd,authblk,bm,dsfont,color,cancel,diagbox,enumerate,epsfig,eucal,fancyhdr,IEEEtrantools,latexsym,lmodern,mathrsfs,mathtools,multicol,multirow,musicography,phaistos,skull,simpler-wick,subcaption,tikz-cd}
\usepackage[normalem]{ulem} 
\usepackage[multiuser]{fixme}
\FXRegisterAuthor{nd}{annd}{Nico} 
\FXRegisterAuthor{cvdv}{ancvdv}{Chris} 
\usetikzlibrary{shapes,snakes} 

\usepackage[
final=true,                    
plainpages=false,              
pdfpagelabels=true,            
pdfencoding=auto,              
unicode=true,                  
hypertexnames=true,            
naturalnames=true              
]{hyperref}
\usepackage[all,cmtip]{xy}
\usetikzlibrary{matrix,calc}
\usepackage[autostyle, english = american]{csquotes}
\MakeOuterQuote{"} 

\definecolor{NiColor}{RGB}{77,77,255}
\definecolor{NiColoRed}{RGB}{255,77,77}
\definecolor{NiCitation}{RGB}{0,181,26}

\newtheoremstyle{TheoremStyle}
{3pt}
{3pt}
{\slshape}
{}
{\sc}
{:}
{.5em}
{}

\theoremstyle{TheoremStyle}
\newtheorem{theorem}{Theorem}
\newtheorem{corollary}[theorem]{Corollary}
\newtheorem{proposition}[theorem]{Proposition}
\newtheorem{lemma}[theorem]{Lemma}

\newtheorem{remark}[theorem]{Remark}
\newtheorem{assumption}{Assumption}

\makeatletter
\def\@endtheorem{\hfill$\lozenge$} 
\makeatother

\title{Lieb-Robinson bounds in classical oscillating lattice systems}
\author[a]{\href{mailto:ian.koot@fau.de}{Ian Koot}}
\author[b]{\href{mailto:chris.ven@fau.de}{C. J. F. van de Ven}}
\affil[a,b]{Friedrich-Alexander-Universit\"{a}t Erlangen-N\"{u}rnberg, Department of Mathematics, Cauerstra\ss e e 11, 91058 Erlangen, Germany}

\def\bearray{\begin{eqnarray}}
\def\earray{\end{eqnarray}}
\def\beq{\begin{equation}}
\def\eeq{\end{equation}}

\def\b0{{\bf 0}}

\def\R{\mathbb R}
\def\Z{\mathbb Z}

\def\S{\mathcal{S}}

\newsymbol\bt 1202             

\def\bR{{\mathbb R}} 
\def\R{{\mathbb R}}

\newsymbol\rest 1316         

\def\b{\textnormal{b}}
\def\CR{C_\mathcal{R}}
\def\SR{\mathcal{S}_\mathcal{R}}

\newcommand{\ran}{\operatorname{ran}}
\newcommand{\pos}{\textnormal{pos}}
\newcommand{\mom}{\textnormal{mom}}

\newcommand{\ignore}[1]{ }
\newcommand{\Hess}{\boldsymbol{\textup{H}}}
\newcommand{\op}{\textup{op}}

\begin{document}

\maketitle
\begin{abstract}
\noindent
The aim of this paper is two-fold. First, we prove the existence of Lieb-Robinson bounds for classical particle systems describing harmonic oscillators interacting with arbitrarily many neighbors, both on lattices and on more general structures. Second, we prove the existence of a global dynamical system on the commutative resolvent algebra,  a C*-algebra of bounded continuous functions on an infinite dimensional vector space, which serves as the classical analog of the Buchholz--Grundling resolvent algebra.
\end{abstract}

\tableofcontents

\section{Introduction}\label{Introduction}
Lieb-Robinson bounds provide a fundamental insight into the effective locality of interactions in quantum and classical many-body systems. Originally established by Lieb and Robinson in the context of quantum spin systems \cite{Lieb_Robinson_1972}, these bounds quantify the maximal speed at which information and correlations can propagate through a lattice of interacting particles. Despite the inherently non-relativistic nature of such systems, Lieb-Robinson bounds imply the existence of an emergent light-cone-like structure, limiting the influence of local perturbations to a finite velocity, commonly referred to as the Lieb-Robinson velocity.

The significance of these bounds extends beyond their conceptual appeal, playing a crucial role in rigorous studies of the dynamics of infinite quantum systems, the stability of topological phases, and the derivation of properties such as exponential clustering of correlations and the existence of thermodynamic limits \cite{Nachtergaele_Sims_2010, Hastings_Koma_2006, Nachtergaele_Sims_2006, Nachtergaele_Sims_Ogata_2006}. More recently, Lieb-Robinson bounds have also been formulated and applied in the analysis of quantum lattice systems with infinite degrees of freedom \cite{Deuchert_Lampart_Lemm_2025,NRSS_2009,Nachtergaele_Sims_Young_2019}, broadening their applicability;  in particular, in the context of resolvent algebras \cite{Deuchert_Lampart_Lemm_2025}.

Analogous bounds also hold in classical many-body systems, where, despite the absence of non-commutativity and unitary evolution, effective locality emerges under suitable conditions.  For example, in Hamiltonian lattice systems with finite-range or rapidly decaying interactions, it is possible to rigorously derive finite propagation speed estimates for perturbations, closely paralleling the quantum setting. Such classical Lieb-Robinson-type bounds have been proven for systems of coupled anharmonic oscillators and classical spins \cite{Marchioro_1978,MBK_2014}, offering insight into the finite-speed propagation and supporting kinetic descriptions in the thermodynamic limit.
\\\\
A particularly important class of models where these ideas apply is given by (an)harmonic lattice Hamiltonians. Here, each lattice site hosts a particle with a continuous degree of freedom, confined in a potential and coupled to its neighbors through linear or nonlinear forces. Such models serve as paradigms for the study of macroscopic non-equilibrium phenomena, such as heat conduction, emerging from many-body Hamiltonian dynamics \cite{Aoki,NRSS_2009}.  

In line with these considerations, this work addresses a broad class of infinite {\em classical} particle systems characterized by general interactions and heterogeneous lattice structures. More precisely,  we consider $\Gamma$ to be a countable index set, which can be interpreted, for example, as a lattice embedded in $\mathbb{R}^\ell$. Around each site $k \in \Gamma$, a particle is confined by a harmonic potential, while particles at different sites interact via pairwise attractive and/or repulsive forces. Denoting by $q_k$ the displacement of the particle at site $k$ from its reference position, and by $p_k$ its momentum, the formal Hamiltonian governing the system is given by
\begin{align}\label{hams}
    H(p,q) = \sum_{k \in \Gamma} \left( \frac{\|p_k\|^2}{2 m_k} + \frac{\nu_k \|q_k\|^2}{2} \right) + \frac{1}{2}\sum<f_{k,l \in \Gamma} V_{kl}(q_k - q_l),
\end{align}

where the mass parameters $m_k > 0$, the harmonic force constants $\nu_k > 0$, and the interaction potentials $V_{kl}$ vary with site indices and satisfy appropriate smoothness and summability conditions (see Section \ref{Sec:oscillating lattice systems} for precise definitions).

Unlike many traditional studies which assume uniform masses and force constants and impose periodic lattice structures, our framework allows for fully inhomogeneous parameters $\{m_k, \nu_k\}$ and considers arbitrary countable sets $\Gamma$ without reliance on any underlying geometric regularity. The thermodynamic limit is formulated purely via the natural partial order on finite subsets of $\Gamma$ under inclusion.  This generality enables our results to apply to a broad spectrum of solid-state material configurations, including but not limited to perfect crystals, amorphous solids such as glass, doped metals with irregular impurities, and novel nanostructures.
\\\\
To analyze locality and propagation properties within this general framework, Section \ref{LR} establishes a Lieb-Robinson-type bound for classical dynamics. Specifically, for any finite sublattice $\Lambda\Subset \Gamma$ of a fixed discrete set $\Gamma$ and any pair of bounded, smooth observables $f$ and $g$ supported on disjoint subsets $X, Y \subset \Lambda$, we show in Theorem \ref{thm: main 1} that the time evolution under the classical Hamiltonian flow  satisfies the estimate
\begin{equation}\label{LB-bounds}
\left\|\left\{\alpha^t_\Lambda(f), g \right\} \right\|_\infty \leq B \|f\|_{C^1} \|g\|_{C^1} \, \mathrm{D}(X, Y) \left( e^{A |t|} - 1 \right),
\end{equation}
where $\alpha^t_\Lambda$ denotes the pullback by the Hamiltonian flow on $\Lambda$, $\{\cdot,\cdot\}$ is the Poisson bracket, $\|\cdot\|_{C^1}$ is the $C^1$-norm, and $A, B > 0$ are constants independent of $f$, $g$, and $\Lambda$. Here,
$$
D(X,Y):=\sum_{x \in X} \sum_{y \in Y} F(d(x,y))
$$
with $F:[0,\infty)\to (0,\infty)$ a suitable function, \textit{cf.} Section \ref{Set-up} for details.  Note that these bounds resemble a classical analog of the usual Lieb-Robinson bounds known in quantum mechanics \cite{NRSS_2009,Nachtergaele_Sims_Ogata_2006,Nachtergaele_Sims_2006,Nachtergaele_Sims_2010}.

Building on this, Theorem \ref{Thm:existence} constitutes our second key result: within the framework of the commutative resolvent algebra introduced in \cite{vanNuland_2019} - a classical analog of  Buccholz-Grundling resolvent algebra \cite{Buchholz_Grundling_2008} - we leverage the Lieb-Robinson bound established in Theorem \ref{thm: main 1} to construct a global $C^*$-dynamical system describing the time evolution of observables in the infinite-volume limit.

This construction highlights the power of our algebraic approach: by capturing finite propagation speeds even in classical systems with potentially non-integrable interactions, it provides a robust foundation for analyzing locality and dynamics in highly general classical settings. Moreover, it opens the door to extending these methods to quantum systems with equally broad and physically realistic interaction structures.

\section{Mathematical setting}\label{Sec:oscillating lattice systems}
We introduce our conventions and assumptions regarding the classical oscillating and interacting (in)finite particle systems under consideration. Fixing some notation, for $F: V \rightarrow W$ a differentiable map between finite-dimensional real inner product spaces, we write $DF(v)$ for the total derivative of $F$ at $v$, meaning that
    \[ \lim_{h \rightarrow 0} \frac{\|F(v+h) - F(v) - DF(v)(h)\|}{\|h\|}=0. \]
We write $L(V,W)$ for the set of all linear maps between $V$ and $W$. For all linear functionals $\varphi\in L(V,\R)$, we write $\varphi^T$ for the element of $V$ that satisfies $\varphi(v) =  \varphi^T \cdot v$ for all $v \in V$. We also write $\nabla f$ for the gradient of a function $f: V\rightarrow \R$, so that $\nabla f(v) = Df(v)^T$.

\subsection{The phase space}\label{sct:subsection The phase space}

We consider an arbitrary countable set $\Gamma$ -- typically interpreted as a discrete subset of $\R^\ell$, namely, as the set of points of confinement around which the particles are pinned by a Harmonic potential. However only the set structure of $\Gamma$ is used, which already endows the set of finite subsets of $\Gamma$ with a partial order -- inclusion -- which is upward directed and hence defines a thermodynamic limit.
Assumptions on the material topology and geometry will be encoded not in $\Gamma$ but in our assumptions on the interaction potentials, to be discussed in \NAK \ref{subsection: Hamiltonian}.

To each element of $\Gamma$ we associate a phase space $\mathbb R^{2d}$, and our total phase space is given by
$$\Omega=\ell_\textnormal{c}(\Gamma,\R^{2d})=\ell_\textnormal{c}(\Gamma,\R^d)\times \ell_\textnormal{c}(\Gamma,\R^d)$$
consisting of pairs $\omega=(p,q)\in\Omega$ of finite sequences $p=(p_l)_{l\in\Gamma},q=(q_l)_{l\in\Gamma}$ for which each entry takes values in $\R^d$.
The components of $p_l$ (resp. $q_l$) in $\R^d$ are denoted $p_{l,i}$ (resp. $q_{l,i}$) for $i=1,\ldots,d$. By construction $\Omega$ is a countably infinite dimensional vector space which admits a natural inner product induced by the inclusion $\Omega\subseteq \ell^2(\Gamma,\R^{2d})$ into the square-summable sequences.

Let $\Lambda\Subset\Gamma$ be any finite subset labeling the particles of a subsystem. We define, for each finite subset $\Lambda\Subset\Gamma$,
$$\Omega_\Lambda:=\{(p,q)\in\Omega:~p_l=q_l=0\text{ for }l\notin\Lambda\}\cong \R^{2|\Lambda|d}.$$
We occasionally adopt the notation
\begin{align*}
    \Omega_\Lambda^\mom=&\{(p,0)\in\Omega_\Lambda\}\cong \R^{|\Lambda|d};\\ 
\Omega_\Lambda^\pos=&\{(0,q)\in\Omega_\Lambda\}\cong \R^{|\Lambda|d},
\end{align*}
and we note that $\Omega_\Lambda=\Omega_\Lambda^\mom\oplus\Omega_\Lambda^\pos$.
We furthermore emphasize that
$$\Omega=\bigcup_{\Lambda\Subset\Gamma}\Omega_\Lambda.$$

\subsection{Local Hamiltonians}\label{subsection: Hamiltonian}
For each finite $\Lambda\Subset\Gamma$ we consider the local Hamiltonian 
\begin{align}\label{Hamiltoniannew2}
H_\Lambda(p,q):=\sum_{k\in\Lambda}\left(\frac{\|p_k\|^2}{2m_k}+\frac{\nu_k\|q_k\|^2}{2}\right)+\frac{1}{2}\sum_{k,l\in\Lambda}V_{kl}(q_k-q_l),
\end{align}
for $(p,q)\in\Omega$. 
Here, $\|\cdot\|$ is the Euclidean norm on $\R^d$,  and $m_k>0$ and $\nu_k> 0$ denote the mass and force constant of particle $k$, and $V_{kl}$ denotes the interaction potential between particles $k$ and $l$, subject to conditions below. 
We note that $H_\Lambda(p,q)$ depends solely on $(p_\Lambda,q_\Lambda)=\pi_\Lambda(p,q)\in\Omega_\Lambda$, and hence we may view $H_\Lambda$ as a function acting on the finite-dimensional phase space $\Omega_\Lambda$.
Observe furthermore that the model defined by \eqref{Hamiltoniannew2} can be interpreted as a generalization of an oscillating and interacting lattice system. 

For a multi-index $\beta:\{1,\cdots, r\}\to\mathbb{Z}_{\geq 0}$ with $|\beta|=\sum_{i=1}^{r}\beta(i)$,  we write
$$\partial^\beta:=\partial_{1}^{\beta(1)}\cdots\partial_{d}^{\beta(r)},$$
where $\partial_{i}^{\beta(i)}=\partial^{\beta(i)}/\partial x_i^{\beta(i)}$ are the usual partial derivatives of order $\beta(i)$ corresponding to the $i^{th}$ coordinate of $\mathbb{R}^d$.

\begin{assumption}\label{conditions}
The following conditions are assumed:
\begin{itemize}
    \item[(i)] $V_{kl}(x)=V_{lk}(-x)$, \quad $V_{kk}(x)=0$;
\item[(ii)] $V_{kl}\in C_0^{\infty}(\mathbb{R}^d,\bR)$ for each $k,l\in\Gamma$;
\item[(iii)] there exists a constant $C_V\geq 0$ and $C_{kl}>0$  for each $k,l\in\Gamma$ such that
$$\|\partial^\beta V_{kl}\|_\infty\leq C_{kl}C_V^{|\beta|}$$
for all $\beta:\{1,\ldots,r\}\to\Z_{\geq 0}$;
\item[(iv)]  $0<\inf_{k\in\Gamma}\{\frac{1}{m_k}\}\leq\sup_{k\in\Gamma}\{\frac{1}{m_k}\}<\infty$, and $0<\inf_{k\in\Gamma}\{\nu_k\}\leq\sup_{k\in\Gamma}\{\nu_k\}<\infty$. 
\end{itemize}
\end{assumption}
It is not difficult to see that the Hamiltonian vector field is \emph{globally Lipschitz} on $\Omega_\Lambda$ \cite[Sec. 2.5]{Nuland_Ven_2023}. Therefore, by the Picard Lindel\"{o}f theorem the Hamiltonian equations admit a unique globally defined solution for every initial condition in $\Omega_\Lambda$. In particular, this guarantees the existence of a continuous Hamiltonian flow
\[
\Phi_t : \Omega_\Lambda \to \Omega_\Lambda,
\]  
which is a homeomorphism for each $t \in \mathbb{R}$.
We furthermore introduce the following notation:
\begin{itemize}\label{notation}
    \item $(p,q)$ $\rightarrow$ arbitrary point in $\Omega_\Lambda$.
    \item $p_j,q_j$ $\rightarrow$ position and momentum component at site $j \in \Lambda$ of the arbitrary point $(p,q) \in \Omega_\Lambda$.
    \item $p(t),q(t)$ $\rightarrow$ the functions that take as an input an initial value (say $(p,q)$), and output the value of the position/momentum at time $t$ (so $\Phi_t(p,q)$. In other words: $p(t) = \pi_{\mom}\circ\Phi_t$ and $q(t) = \pi_{\pos}\circ \Phi_t$.
    \item $Q(t),P(t)$ $\rightarrow$ $Q(t) = q(t)(P(0),Q(0))$ and $P(t) = p(t)(P(0),Q(0))$ for some fixed choice of $P(0),Q(0) \in \Omega_\Lambda$.
\end{itemize}


\subsection{Conditions on $\Gamma$}\label{Set-up}
Let $\Gamma$ be a countable metric space equipped with a metric $d$.  We restrict the geometry of $\Gamma$ by assuming the existence of a function
$$
F: [0,\infty)\to (0,\infty),
$$
with the following properties (see e.g. \cite{NRSS_2009,Nachtergaele_Sims_Ogata_2006}):
\begin{assumption}
We assume:
    \begin{enumerate}
    \item \textit{Monotonicity:} $F$ is non-increasing in its argument, i.e.,
    $$
    F(x) \leq F(y) \quad \text{for all } y \leq x,
    $$
    and normalized such that
    \[
    F(0) = 1.
    \]
\item \textit{Uniform integrability:} For every fixed $y \in \Gamma$, the sum over all $x \in \Gamma$ of $F(d(x,y))$ is finite:
    $$
    \sum_{x \in \Gamma} F(d(x,y)) < \infty,
    $$
    and
    $$
    \|F\| := \sup_{y \in \Gamma} \sum_{x \in \Gamma} F(d(x,y)) < \infty.
    $$
\item \textit{Convolution property:} There exists a constant $C_F > 0$ such that for all $x,y \in \Gamma$,
    $$
    \sum_{z \in \Gamma} F(d(x,z)) F(d(z,y)) \leq C_F F(d(x,y)).
    $$
\end{enumerate}
\end{assumption}
Consider the potential function $\Psi$ defined on two point sets $Z=\{k,l\}$ by $\Psi(Z):=V_{kl}:\mathbb{R}^d\to\mathbb{R}$. It follows that
$$H_\Lambda=\sum_{k\in\Lambda}H_k^0+\sum_{\substack{Z\subset\Lambda\\ |Z|=2}}\Psi(Z),$$
where $$H_k^0=\frac{\|p_k\|^2}{2m_k}+\frac{\nu_k\|q_k\|^2}{2}.$$
Finally, we assume the following compatibility condition between the grid $\Gamma$ and the interaction part of the Hamiltonian:
\begin{assumption}\label{ass:psi}
    The following constant is finite:
$$
\|\Psi\|:=\sup_{k,l \in \Gamma}\frac{C_{kl}}{F(d(k,l))}<\infty,$$
where the $C_{kl}$ are the constants from Assumption (iii). In particular, for all $k,l\in\Gamma$
$$C_{kl}\leq \|\Psi\|F(d(k,l)).
$$

\end{assumption}


\section{Lieb-Robinson bounds}\label{LR}

\subsection{Estimates on time-evolved Poisson bracket}\label{sec:estimates}
Building on the methods developed in \cite{Marchioro_1978}, we generalize them to our framework. Let $X,Y \subset \Lambda$ be finite subsets, and $f_0 \in C_b^1(\Omega_X)$ and $g_0 \in C_b^1(\Omega_Y)$. In order to compare $f_0$ and $g_0$, in accordance with the structure of the commutative resolvent algebra (see Section \ref{Commutative resolvent algebra}) we define $\pi_{
X,\Lambda}: \Omega_\Lambda \rightarrow \Omega_X$ and $\pi_{Y,\Lambda}: \Omega_\Lambda \rightarrow \Omega_Y$ to be the orthogonal projection from $\Omega_\Lambda$ onto $\Omega_X$ and $\Omega_Y$, respectively (recall that all these spaces are finite-dimensional). We then write $f := f_0\circ\pi_{X,\Lambda}$ and $g:= g_0\circ\pi_{Y,\Lambda}$. Their $ C^1 $-norms are defined as:
$$
\|f\|_{C^1} :=\|f\|_\infty + \| \nabla f\|_{2,\infty}, \quad
\|g\|_{C^1} :=  \|g\|_\infty + \| \nabla g\|_{2,\infty},
$$
the notation $\|\nabla f\|_{2,\infty}$ stands for $\sup_{(p,q) \in \Omega_\Lambda} \|\nabla f(p,q)\|_2 $ with $\|\cdot\|_2$ the standard Euclidean norm on $\Omega_\Lambda$ induced by $\ell^2(\Gamma,\R^{2d})$, and similarly for $ \| \nabla g\|_{2,\infty}$. In other words, we have for $(p,q)\in \Omega_\Lambda$ that
    \[ \|(p,q)\|_2 = \sqrt{\sum_{j\in \Lambda}\|p_j\|^2 + \|q_j\|^2} =\sqrt{\sum_{j\in\Lambda}\sum_{i=1}^d p_{j,i}^2 + q_{j,i}^2 }.\]
For fixed lattice points $j \in \Lambda$ with $(p_j,q_j)\in\mathbb{R}^{2d}$, we simply write $$
\|(p_j,q_j)\|=\sqrt{\sum_{j=1}^
d p_{j,i}^2+q_{j,i}^2},
$$
without explicitly indicating the subscript $2$.
The Poisson bracket of the time-evolved observable $\alpha_t(f):=f\circ \Phi_t$ with $g$ is given by
$$
\{\alpha_t(f), g\} = \sum_{j \in \Lambda} \sum_{i=1}^d \left( 
\frac{\partial \alpha_t(f)}{\partial q_{j,i}}  \frac{\partial g}{\partial p_{j,i}} 
- \frac{\partial \alpha_t(f)}{\partial p_{j,i}}   \frac{\partial g}{\partial q_{j,i}} \right).
$$ 
We will bound this quantity by bounding on the one hand the derivatives of $f$ and $g$, and on the other hand the dependence of the time evolution on the initial conditions. To aid with notation, for all $k \in \Lambda$ let us define the projections 
\begin{IEEEeqnarray*}{rClrCl}
    \pi_{\mathrm{pos}}(p,q) & := & q \in \Omega_\Lambda ^{\textup{pos}} \cong \R^{|\Lambda|d} \quad & \pi_{\mathrm{pos},k}(p,q) &:= &q_k \in \Omega_{\{k\}}^{\pos} \cong \R^{d} \\
    \pi_{\mathrm{mom}}(p,q) &:=& p \in \Omega_{\Lambda}^{\textup{mom}} \cong \R^{|\Lambda|d} \quad &  \pi_{\mathrm{mom},k}(p,q) &:=& p_k \in \Omega_{\{k\}}^\mom \cong \R^d. 
\end{IEEEeqnarray*}
where $(p,q) \in \Omega_\Lambda^\mom \oplus \Omega_\Lambda^\pos = \Omega_\Lambda$. We will write $\pi_{\mathrm{pos}}^*: \Omega_\Lambda^{\textup{pos}} \rightarrow \Omega_\Lambda$ etc. for the conjugate maps, i.e. the associated inclusion. We then define for all functions $h \in C^1_b(\Omega_\Lambda,V)$ (where $V$ is a finite dimensional inner product space) the generalized partial derivatives
    \begin{align*} \frac{\partial h}{\partial q}(p,q) & := D h(p,q) \circ \pi_{\textup{pos}}^* :\Omega_\Lambda^{\textup{pos}} \rightarrow V \\
    \frac{\partial h}{\partial p}(p,q) & := D h(p,q) \circ \pi_{\textup{mom}}^*: \Omega_\Lambda^{\textup{mom}} \rightarrow V\end{align*} 
and also for all $k\in \Lambda$ the site-specific derivatives
\begin{align*} \frac{\partial h}{\partial q_k}(p,q) & := D h(p,q) \circ \pi_{\textup{pos},k}^* :\Omega_{\{k\}}^{\textup{pos}} \rightarrow V \\
    \frac{\partial h}{\partial p_k}(p,q) & := D h(p,q) \circ \pi_{\textup{mom},k}^*: \Omega_{\{k\}}^{\textup{mom}} \rightarrow V\end{align*} 
We note that if $V= \R$ then 
\[ \left(\frac{\partial h}{\partial q}(p,q) \right)^T = \pi_{\pos} (Dh(p,q))^T = \pi_\pos \nabla h(p,q) \]
and similarly for $\frac{\partial h}{\partial p}(p,q)$, $\frac{\partial h}{\partial q_k}(p,q)$, and $\frac{\partial h}{\partial p_k}(p,q)$ for all $k \in \Lambda$. We can then succinctly write
    \[ \{ \alpha_t(f),g\} =  
\left(\frac{\partial \alpha_t(f)}{\partial q}\right)^T \cdot \left( \frac{\partial g}{\partial p}\right)^{T}  - \left(\frac{\partial \alpha_t(f)}{\partial p} \right)^T \cdot \left( \frac{\partial g}{\partial q}\right)^{T} \]
where the inner product and transpose are taken pointwise.

Since $\alpha_t(f) = f \circ \Phi_t$, we have 
    \begin{align*} D(\alpha_t(f))(p,q)  = & \, D(f)(\Phi_t(p,q)) \circ D\Phi_t(p,q) \\
    = & \, \frac{\partial f}{\partial q}(\Phi_t(p,q)) \circ \pi_{\textup{pos}} \circ D \Phi_t(p,q) \\
    &+ \frac{\partial f}{\partial p}(\Phi_t(p,q)) \circ \pi_{\textup{mom}} \circ D \Phi_t
(p,q)    \end{align*}
This leads us to define the functions
    \begin{align*} q(t) & := \pi_\pos\circ\Phi_t: \Omega_\Lambda \rightarrow \Omega_\Lambda^\pos \\
    p(t) &:= \pi_\mom \circ \Phi_t: \Omega_\Lambda \rightarrow \Omega_\Lambda^\mom \\
    q_k(t) &:= \pi_{\pos,k} \circ \Phi_t : \Omega_\Lambda \rightarrow \Omega_{\{k\}}^\pos \\
    p_k(t)& := \pi_{\mom,k} \circ \Phi_t:\Omega_\Lambda \rightarrow \Omega_{\{k\}}^\mom
    \end{align*}
mapping the initial conditions to various quantities after a time evolution by time $t$. Since projections are linear, we calculate
    \begin{align*} \frac{\partial q(t)}{\partial q}(p,q) & = \pi_{\textup{pos}} \circ D \Phi_t(p,q) \circ \pi_{\textup{pos}}^*  \\
     \frac{\partial q(t)}{\partial p}(p,q) & := \pi_{\textup{pos}} \circ D \Phi_t(p,q) \circ \pi_{\textup{mom}}^* 
     \end{align*}
and similarly for $\frac{\partial p(t)}{\partial q}(p,q)$ and $\frac{\partial p(t)}{\partial p}(p,q)$. We therefore see that
\begin{align*} \frac{\partial \alpha_t(f)}{\partial q}(p,q) & = \frac{\partial f}{\partial q}(\Phi_t(p,q)) \circ \frac{\partial q(t)}{\partial q}(p,q) + \frac{\partial f}{\partial p}(\Phi_t(p,q)) \circ \frac{\partial p(t)}{\partial q}(p,q) \\
\frac{\partial \alpha_t(f)}{\partial p}(p,q) & = \frac{\partial f}{\partial q}(\Phi_t(p,q)) \circ \frac{\partial q(t)}{\partial p}(p,q) + \frac{\partial f}{\partial p}(\Phi_t(p,q)) \circ \frac{\partial p(t)}{\partial p}(p,q) 
\end{align*}
Summarizing and using the fact that $(\varphi\circ A)^T \cdot v =  \varphi^T \cdot Av$ for a linear functional $\varphi:W \rightarrow \R$ and linear map $A: V \rightarrow W$, we obtain
\begin{align*}
    \{\alpha_t(f),g\}(p,q) = {}& \left(\frac{\partial f}{\partial q}(\Phi_t(p,q))\right)^T  \cdot  \frac{\partial q(t)}{\partial q}(p,q)\left(\frac{\partial g}{\partial p}(p,q) \right)^T \\
    & +  \left(\frac{\partial f}{\partial p}(\Phi_t(p,q))\right)^T  \cdot \frac{\partial p(t)}{\partial q}(p,q) \left(\frac{\partial g}{\partial p}(p,q) \right)^T \\
    &- \left(\frac{\partial f}{\partial q}(\Phi_t(p,q))\right)^T\cdot  \frac{\partial q(t)}{\partial p}(p,q)\left(\frac{\partial g}{\partial q}(p,q) \right)^T \\
    & - \left(\frac{\partial f}{\partial p}(\Phi_t(p,q)) \right)^T\cdot \frac{\partial p(t)}{\partial p}(p,q) \left(\frac{\partial g}{\partial q}(p,q) \right)^T
\end{align*}
Recall that $f$ only depends on the grid sites in $X$ and $g$ on the grid sites in $Y$, by which we mean that $f = f_0\circ\pi_{X,\Lambda}$ and $g = g_0\circ\pi_{Y,\Lambda}$. This means that $\frac{\partial f}{\partial q_j} = \frac{\partial f}{\partial p_j} = 0$ for $j \in \Lambda \setminus X$, as well as $\frac{\partial g}{\partial q_j} = \frac{\partial g}{\partial p_j} = 0$ for $j \in \Lambda \setminus Y$. We therefore see that
\begin{align*}
    \{\alpha_t(f),g\}(p,q) = {}& \sum_{j\in X} \sum_{k \in Y} \left(\frac{\partial f}{\partial q_j}(\Phi_t(p,q)) \right)^T \cdot \frac{\partial q_j(t)}{\partial q_k}(p,q)  \left(\frac{\partial g}{\partial p_k}(p,q) \right)^T \\
    & + \left(\frac{\partial f}{\partial p_j}(\Phi_t(p,q))\right)^T \cdot \frac{\partial p_j(t)}{\partial q_k}(p,q) \left(\frac{\partial g}{\partial p_k}(p,q) \right)^T \\
    &- \left(\frac{\partial f}{\partial q_j}(\Phi_t(p,q))\right)^T \cdot \frac{\partial q_j(t)}{\partial p_k}(p,q)  \left(\frac{\partial g}{\partial q_k}(p,q) \right)^T \\
    & - \left(\frac{\partial f}{\partial p_j}(\Phi_t(p,q))\right)^T \cdot \frac{\partial p_j(t)}{\partial p_k}(p,q)\left(\frac{\partial g}{\partial q_k}(p,q) \right)^T
\end{align*}

Because $(\frac{\partial f}{\partial q_j}(p,q))^T$ corresponds to projecting $\nabla f(p,q)$ onto $\Omega_{\{j\}}^{\textup{pos}}$, we have that $\|\frac{\partial f}{\partial q_j}(p,q)\|_2 \leq \|\nabla f(p,q)\|_2$ for all $j \in \Lambda$.  Using that for all $(p,q) \in \Omega_\Lambda$ we have $\|\nabla f(\Phi_t(p,q))\|_2 \leq \|\nabla f\|_{2,\infty}$ and $\|\nabla g(p,q)\|_2 \leq \|\nabla g\|_{2,\infty}$ we then estimate
\begin{align*} & \|\{\alpha_t(f),g\}\|_\infty \\ & \leq \|f\|_{C^1}\|g\|_{C^1} \sum_{\substack{j \in X \\ k \in Y}} \left( \left\| \frac{\partial q_j(t)}{\partial q_k}\right\|_{\op,\infty} + \left\| \frac{\partial p_j(t)}{\partial q_k}\right\|_{\op,\infty} +\left\| \frac{\partial q_j(t)}{\partial p_k}\right\|_{\op,\infty} + \left\| \frac{\partial p_j(t)}{\partial p_k}\right\|_{\op,\infty} \right) \\
& \leq 4\|f\|_{C^1}\|g\|_{C^1} \sup_{\substack{j \in X \\ k \in Y}} \left\{ \left\| \frac{\partial q_j(t)}{\partial q_k}\right\|_{\op,\infty}, \left\| \frac{\partial p_j(t)}{\partial q_k}\right\|_{\op,\infty}, \left\| \frac{\partial q_j(t)}{\partial p_k}\right\|_{\op,\infty}, \left\| \frac{\partial p_j(t)}{\partial p_k}\right\|_{\op,\infty} \right\},
\end{align*}
where for a function $F: \Omega_\Lambda \rightarrow L(V,W)$ with values in the set of linear operators between two vector spaces $V$ and $W$ we write
    \[ \|F\|_{\op,\infty} = \sup_{(p,q)\in \Omega_\Lambda} \|F(p,q)\|_{\op} \]

\begin{remark}
    Note that whenever $X$ and $Y$ are disjoint, their are no diagonal terms. This will play a role in the next section.
\end{remark}

\subsection*{Variational methods}
Now that we have separated the dependence on the functions $f$ and $g$, we focus on bounding the derivatives of the time evolution with respect to the initial data. In order to clean up notation, we pick an orbit of the time evolution $(P(t),Q(t)):= \Phi_t(P(0),Q(0))$ for some fixed initial point $(P(0),Q(0)) 
\in \Omega_\Lambda$ (we write uppercase letters to distinguish for example $Q(t)$, which for all $t \in \R$ is a vector in $\Omega_{\Lambda}^\pos$, from $q(t)$, which for all $t \in \R$ is a function $\Omega_\Lambda\rightarrow \Omega_\Lambda^\pos$ describing how the position at time $t$ depends on the initial value; in other words, $Q(t) = q(t)(P(0),Q(0))$).
%
We write the following for the four Jacobian matrices
\begin{IEEEeqnarray*}{rClrCl}
    X_{kj}(t) &:= & \frac{\partial q_k(t)}{\partial q_j}(P(0),Q(0)); \quad \quad&
    Z_{kj}(t) &:= &\frac{\partial q_k(t)}{\partial p_j}(P(0),Q(0)); \\
    Y_{kj}(t) &:= &\frac{\partial p_k(t)}{\partial q_j}(P(0),Q(0)); \quad \quad &
    W_{kj}(t) &:= & \frac{\partial p_k(t)}{\partial p_j}(P(0),Q(0));
\end{IEEEeqnarray*}
where we remind the reader that for each $k,j \in \Lambda$ these are each a $d \times d$ matrix.

Recall the Hamiltonian of the system:
$$
H_\Lambda(p,q) = \sum_{k \in \Lambda} \left( \frac{\|p_k\|^2}{2 m_k} + \frac{\nu_k}{2} \|q_k\|^2 \right) + \frac{1}{2}\sum_{k,l \in \Lambda} V_{kl}(q_k - q_l),
$$
where the potentials \(V_{kl}\) satisfy Assumption \ref{conditions}. The evolution $\Phi_t(p(0),q(0)):=(q(t), p(t))$ is given by solving Hamilton's equations, which are typically nonlinear due to the interaction potentials $V_{kl}$. The linearized dynamics come from differentiating the nonlinear flow with respect to initial conditions. 
To analyze these, we consider the Hessian matrix of the potential energy evaluated along the trajectory $(q(t), p(t))$. So for $h: V \rightarrow \R$ differentiable such that $\nabla h$ is differentiable, we write
    \[ \Hess(h)(v):= D(\nabla h)(v) \in L(V,V).  \]
for the Hessian at the point $v \in V$. We also want to formulate mixed second partial derivatives with respect to the vector variables $q_k$, so for $U: \Omega_\Lambda^{\textup{pos}} \rightarrow \R$ we define
    \[ \frac{\partial^2 U}{\partial q_k \partial q_j}(q) :=  D\left(\left(\frac{\partial U}{\partial q_j} \right)^T\right) (q)\circ \pi_{\textup{pos},k}^* \in L(\Omega_{\{k\}}^{\textup{pos}},\Omega_{\{j\}}^{\textup{pos}}).\]
However, because
    \[ \left( \frac{\partial U}{\partial q_j}\right)^T(q) = (DU(q) \circ\pi_{\pos,j}^*)^T = \pi_{\pos,j}(DU(q))^T \]
we have
    \[ \frac{\partial^2 U}{\partial q_k \partial q_j}(q) = \pi_{\pos,j} \circ \Hess(U)(q) \circ \pi_{\pos,k}^* \]

\begin{lemma}\label{Lemm: Hessian}
Let $Q(t)$ be a path in $\Omega_\Lambda^\pos$. The  Hessian matrix of the potential energy part of the Hamiltonian
$$
B(t) := \left[ B_{kj}(t) \right]_{j,k \in \Lambda}, \quad \text{where} \quad B_{kj}(t) := \frac{\partial^2 U}{\partial q_k \partial q_j}(Q(t))
$$
with 
$$
U(q) := \sum_{k \in \Lambda} \frac{\nu_k}{2} \|q_k\|^2 + \frac{1}{2}\sum_{k,l \in \Lambda} V_{kl}(q_k - q_l)
$$
has the following block structure
\begin{align*}
B_{jj}(t) & = \nu_j  + \sum_{l \in \Lambda} \Hess( V_{jl} )(Q_j(t) - Q_l(t)) \\
B_{kj}(t) & = - \Hess( V_{kj})(Q_k(t) - Q_j(t)) \quad (j \neq k).
\end{align*}
In particular, we have $B_{kj}(t) = B_{jk}(t) = B_{jk}(t)^T$.
\end{lemma}

\begin{proof}
We first calculate the first partial derivatives of $U$: the total derivative equals
    \begin{align*}
        DU(q) (v) & = \sum_{k \in \Lambda} \nu_k  q_k\cdot \pi_{\pos,k}v \\
        & \hphantom{ABCDE} + \frac{1}{2} \sum_{k,l \in \Lambda} D(V_{kl})(q_k - q_l) \circ D(\pi_{\pos,k} - \pi_{\pos,l})(v) \\
        & = \sum_{k\in \Lambda} \nu_k q_k \cdot \pi_{\pos,k}v + \frac{1}{2}\sum_{k,l\in\Lambda} D(V_{kl})(q_k - q_l)(\pi_{\pos,k}v - \pi_{\pos,l}v)
    \end{align*}
so that  
\begin{align*}
    DU(q) \circ \pi_{\pos,j}^*(v) = \nu_jq_j \cdot v  + \frac{1}{2} \sum_{k,l \in \Lambda}(\delta_{k,j} - \delta_{l,j})D(V_{kl})(q_k - q_l
)(v) \end{align*}
meaning that
\begin{align*} \left(\frac{\partial U}{\partial q_j}(q)\right)^T & = \nu_jq_j + \frac{1}{2}  \sum_{l \in \Lambda} \nabla V_{jl}(q_j - q_l) - \nabla V_{lj}(q_l - q_j) \\
&= \nu_jq_j + \sum_{l \in \Lambda}\nabla V_{jl}(q_j - q_l)
\end{align*}
Calculating the full derivative of this, we see that
\begin{align*} 
D\left( \left( \frac{\partial U}{\partial q_j}\right)^T\right) (q)& = \nu_j \pi_{\pos,j} + \sum_{l \in \Lambda} D(\nabla V_{jl} )(q_j - q_l) \circ D(\pi_{\pos,j} - \pi_{\pos,l}) \\
&= \nu_j \pi_{\pos,j} + \sum_{l \in \Lambda} \Hess(V_{jl})(q_j - q_l)\circ(\pi_{\pos,j} - \pi_{\pos,l}).
\end{align*}

\textbf{Diagonal blocks:} For $j = k$, we compute
\begin{align*}
B_{jj}(t) & = \frac{\partial^2 U}{\partial q_j^2}(Q(t)) \\
& = D\left( \left( \frac{\partial U}{\partial q_j}\right)^T\right) (Q(t)) \circ \pi_{\pos,j}^* \\
& = \nu_j I_d + \sum_{l \in \Lambda} \Hess(V_{jl})(Q_j(t) - Q_l(t))
\end{align*}
where in the last step we used that $\pi_{\pos,l} \circ \pi_{\pos,j}^* = \delta_{l,j} I_d$ and $V_{jj} = 0$.

\textbf{Off-diagonal blocks:} For $j \neq k$, we calculate
\begin{align*}
B_{kj}(t) & = \frac{\partial^2 U}{\partial q_k \partial q_j}(Q(t)) \\
& = D\left( \left( \frac{\partial U}{\partial q_j}\right)^T\right) (Q(t)) \circ \pi_{\pos,k}^* \\
& = - \Hess(V_{jk})(Q_j(t) - Q_k(t)) \\
& = - \Hess(V_{kj})(Q_k(t) - Q_j(t))
\end{align*}
where in the last step we used that $V_{jk}(x) = V_{kj}(-x)$, so $\nabla V_{jk}(x)= -\nabla V_{kj}(-x)$ and $\Hess(V_{jk})(x) = \Hess(V_{kj})(-x)$. This also shows that $B_{kj}(t) = B_{jk}(t)$, and since the Hessian of a smooth function is symmetric, we indeed see that $B_{kj}(t) = B_{jk}(t)^T$, as the identity $\frac{\partial^2 U}{\partial q_k \partial q_j}(q) = \pi_{\pos,j} \circ \Hess(U)(q) \circ \pi_{\pos,k}^*$ suggests.
\end{proof}

Next, we again consider the Jacobian blocks describing the sensitivity of the flow with respect to initial conditions. The position and momentum satisfy the Hamiltonian differential equations, and as we show in the next Lemma, because everything is sufficiently smooth, we can permute the time derivative and the derivative in the initial data to arrive at differential equations for the derivatives of position and momentum with respect to the initial data.
\begin{lemma}\label{lem:variational_jacobian}
Let $(P(t), Q(t))$ be solutions of the Hamiltonian system
$$
H(q,p) = \sum_{j \in \Lambda} \frac{1}{2m_j} \|p_j\|^2 + \sum_{j \in \Lambda} \frac{\nu_j}{2} \|q_j\|^2 + \frac{1}{2} \sum_{j,k \in \Lambda} V_{jk}(q_j - q_k).
$$
For each $k,j \in \Lambda$ we define
\begin{IEEEeqnarray*}{rClrCl}
    X_{kj}(t) &:= & \frac{\partial q_k(t)}{\partial q_j}(P(0),Q(0)); \quad \quad&
    Z_{kj}(t) &:= &\frac{\partial q_k(t)}{\partial p_j}(P(0),Q(0)); \\
    Y_{kj}(t) &:= &\frac{\partial p_k(t)}{\partial q_j}(P(0),Q(0)); \quad \quad &
    W_{kj}(t) &:= & \frac{\partial p_k(t)}{\partial p_j}(P(0),Q(0));
\end{IEEEeqnarray*}
Then the variational system satisfies
$$
\begin{cases}
\dot{X}_{kj}(t) = \frac{1}{m_k}Y_{kj}(t), \\
\dot{Y}_{kj}(t) =  -\sum_{l \in \Lambda} B_{lk}(t) X_{lj}(t), \\
\dot{Z}_{kj}(t) = \frac{1}{m_k} W_{kj}(t), \\
\dot{W}_{kj}(t) =  -\sum_{l\in\Lambda}B_{lk}(t) Z_{lj}(t),
\end{cases}
\quad
\text{with} \quad \begin{cases}
X_{kj}(0) = \delta_{kj}I_d \\
Y_{kj}(0) = 0 \\
Z_{kj}(0) = 0 \\
W_{kj}(0) = \delta_{kj}I_d
\end{cases}
$$
where $B_{kj}(t)$ is the Hessian of the potential part of $H$ with respect to positions $k$ and $j$, i.e.
$$
B_{kj}(t)= \frac{\partial^2 U}{\partial q_k \partial q_j} (Q(t)),
$$
whose matrix elements are explicitly given by the formulas obtained from Lemma \ref{Lemm: Hessian}.
\end{lemma}

\begin{proof}
The initial values are easy to check because $\Phi_0$ is the identity, so
    \[ \frac{\partial q_k(0)}{\partial q_j} = \pi_{\pos,k}\circ D(\Phi_0) \circ \pi_{\pos,j}^* = \pi_{\pos,k} \circ \pi_{\pos,j}^* = \delta_{k,j} I_d  \]
and similar for the rest. 

To verify the system of differential equations, we define $\Phi:\R\times\Omega_\Lambda \rightarrow \Omega_\Lambda$ by $\Phi(t,p,q) := \Phi_t(p,q)$. We are looking to express
    \begin{align*} \frac{d X_{k,j}}{dt}(t) & = \frac{d}{dt}\left( \frac{\partial q_k(\cdot)}{\partial q_j}(P(0),Q(0))\right)(t) \\
    & =  \frac{d}{dt}\left( \frac{\partial \pi_{\pos,k}\circ\Phi}{\partial q_j}(\cdot, P(0),Q(0)) \right)(t) \\
    & = \frac{\partial}{\partial t} \left(\frac{\partial \pi_{\pos,k}\circ\Phi}{\partial q_j} \right)(t,P(0),Q(0))
    \end{align*}
(here we have used the notation $\frac{\partial q_k(\cdot)}{\partial q_j}(P(0),Q(0))$ for the map $t \mapsto {\frac{\partial q_k(t)}{\partial q_j}(P(0),Q(0))} $ and similarly in the second line). Since 
    \[ \frac{\partial \pi_{\pos,k}\circ\Phi}{\partial q_j} : \R \times \Omega_\Lambda \rightarrow L(\Omega_{\{j\}}^\pos,\Omega_{\{k\}}^\pos )\]
is linear-operator valued function in finite dimensions, we can consider the partial derivative in a strong sense; it therefore suffices to evaluate the partial derivative of $\frac{\partial \pi_{\pos,k}\circ\Phi}{\partial q_{j,i}}$ for all $i = 1,\ldots, d$, which are $\frac{\partial \pi_{\pos,k}\circ\Phi}{\partial q_j}$ applied to the $i$'th unit vector in $\Omega_{\{j\}}^\pos$.

We note that $\frac{\partial H}{\partial p_{k,i}}: \Omega_\Lambda \rightarrow \R$ and $\frac{\partial H}{\partial q_{k,i}}:\Omega_\Lambda \rightarrow \R$ are smooth for all $k \in \Lambda$ and $i = 1, \ldots ,d$, so $\Phi$ is the solution to an ordinary differential equation whose time derivative at time $t$ is a smooth function of the value at time $t$; by standard results, this means that $\Phi$ is a smooth map (see e.g. \cite[Sec 32.4]{ArnoldODE}). Since this in particular means that it is $C^2$, we know that its derivatives commute, and therefore
    \begin{align*} \frac{\partial^2 \pi_{\pos,k} \circ \Phi}{\partial t\partial q_{j,i}}(t,P(0),Q(0)) & = \frac{\partial^2 \pi_{\pos,k} \circ \Phi}{\partial q_{j,i}\partial t}(t,P(0),Q(0)) \\
    & = \frac{\partial}{\partial q_{j,i}}\left( \frac{\partial \pi_{\pos,k} \circ \Phi}{\partial t} (t, \cdot,\cdot)\right)(P(0),Q(0))
    \end{align*}
(see e.g. \cite[Prop. C.6]{LeeManifolds}). By Hamilton's equations, we know that 
    \[ \frac{\partial Q_{k,i}}{\partial t}(t) = \frac{\partial H}{\partial p_{k,i}}(P(t),Q(t)) \quad \frac{\partial P_{k,i}}{\partial t}(t) = - \frac{\partial H}{\partial q_{k,i}}(P(t),Q(t)) \]
for all trajectories $(P(t),Q(t)) = \Phi_t(P(0),Q(0))$, $k \in \Lambda$ and $i = 1, \ldots d$. We can reformulate this as
    \[ \frac{\partial \pi_{\pos,k} \circ \Phi}{\partial t} (t, p,q) = \frac{d ( q_k(\cdot)(p,q))}{d t}(t) =  \left(\frac{\partial H}{\partial p_k}(\Phi_t(p,q)) \right)^T\]
meaning that
    \[ \frac{\partial^2 \pi_{\pos,k} \circ \Phi}{\partial t\partial q_{j,i}}(t,P(0),Q(0)) = \frac{\partial}{\partial q_{j,i}} \left( \left(\frac{\partial H}{\partial p_k}\right)^T \circ \Phi_t \right)(P(0),Q(0)). \]
So we conclude that
    \begin{align*} \frac{d X_{kj}}{dt}(t) & = \frac{\partial }{\partial q_j} \left( \left(\frac{\partial H}{\partial p_k}\right)^T \circ \Phi_t \right)(P(0),Q(0)) \\
    & = \sum_{l \in \Lambda}\frac{\partial^2 H}{\partial q_l \partial p_k}(P(t),Q(t)) \circ X_{lj}(t)  + \frac{\partial^2 H}{\partial p_l \partial p_k}(P(t),Q(t)) \circ Y_{lj}(t) 
    \end{align*}
By similar calculations we see that
\begin{align*}
\frac{d Y_{kj}}{dt}(t)
    & = -\sum_{l \in \Lambda}\frac{\partial^2 H}{\partial q_l \partial q_k}(P(t),Q(t)) \circ X_{lj}(t)  + \frac{\partial^2 H}{\partial p_l \partial q_k}(P(t),Q(t)) \circ Y_{lj}(t) \\
\frac{d Z_{kj}}{dt}(t)
    & = \sum_{l \in \Lambda}\frac{\partial^2 H}{\partial q_l \partial p_k}(P(t),Q(t)) \circ Z_{lj}(t)  + \frac{\partial^2 H}{\partial p_l \partial p_k}(P(t),Q(t)) \circ W_{lj}(t) \\
\frac{d W_{kj}}{dt}(t)
    & = -\sum_{l \in \Lambda}\frac{\partial^2 H}{\partial q_l \partial q_k}(P(t),Q(t)) \circ Z_{lj}(t)  + \frac{\partial^2 H}{\partial p_l \partial q_k}(P(t),Q(t)) \circ W_{lj}(t)
\end{align*}
Since our Hamiltonian doesn't contain any mixed position and momentum terms, the mixed position-momentum derivatives vanish. One easily calculates that
    \[ \left(\frac{\partial \|p_k\|^2}{\partial p_k}\right)^T = 2\pi_{\mom,k},\quad\textup{so} \quad  \frac{\partial^2 \|p_k\|^2}{\partial p_l \partial p_k} = 2\pi_{\mom,k} \circ\pi_{\mom,l}^*.\]
We can therefore reduce the system of differential equations above to
\begin{align*}
\frac{d X_{kj}}{dt}(t)
    & = \frac{1}{m_k} Y_{kj}(t) \\
\frac{d Y_{kj}}{dt}(t)
    & = -\sum_{l \in \Lambda}\frac{\partial^2 U}{\partial q_l \partial q_k}(P(t),Q(t)) \circ X_{lj}(t) \\
\frac{d Z_{kj}}{dt}(t)
    & = \frac{1}{m_k} W_{kj}(t) \\
\frac{d W_{kj}}{dt}(t)
    & = -\sum_{l \in \Lambda}\frac{\partial^2 U}{\partial q_l \partial q_k}(P(t),Q(t)) \circ Z_{lj}(t)
\end{align*}
where $U(p,q) = H(p,q) - \sum_{k\in \Lambda} \|p_k\|^2$ is the potential function. By Lemma \ref{Lemm: Hessian} the result now follows.
\end{proof}

\begin{lemma}\label{lemma: uniqueness}
    The solutions to the ODE's defined in Lemma \ref{lem:variational_jacobian} are all uniquely defined.
\end{lemma}
\begin{proof}
Fix $\Lambda \Subset \Gamma$ finite. For each pair $j,k \in \Lambda$ consider the variational blocks
$$
X_{kj}(t), Y_{kj}(t), Z_{kj}(t), W_{kj}(t) \in \mathbb{R}^{d \times d}
$$
satisfying the system of ODEs as described in Lemma \ref{lem:variational_jacobian}.

By Assumption \eqref{conditions}, and the fact that the flow $\Phi(t)$ is continuous, the matrix coefficients 
$B_{kj}(t)$ are continuous in $t$.
In particular, 
$$
\| B_{kj}(t) \|_{\infty} \;\le\;
\begin{cases}
\nu_j + \displaystyle\sum_{l\in\Lambda} \| \Hess( V_{jl} )\|_{\op,\infty}, & j = k, \\[1em]
\| \Hess( V_{kj} )\|_{\op,\infty}, & j \neq k,
\end{cases}
$$
and the right-hand sides are finite by the decay assumptions on partial derivatives of $V_{kj}$ for all $k,j\in\Lambda$. 
Hence the map 
$t\mapsto B_{kj}(t)$ 
is bounded and continuous.

We now stack the unknowns into
$$
T_{kj}(t) = \big(X_{k j}(t), Y_{kj}(t), Z_{k j}(t), W_{k j}(t)\big)^\top.
$$
Then
$$
\dot T_{kj}(t) =A_{kj}(t) T_{kj}(t), \quad
A_{kj}(t) =
\begin{pmatrix}
0 & M_{kj}^{-1} & 0 & 0\\
-B_{kj}(t) & 0 & 0 & 0\\
0 & 0 & 0 & M_{kj}^{-1}\\
0 & 0 & -B_{kj}(t) & 0
\end{pmatrix},
$$
where $M_{kj}^{-1} = \delta_{kj}\text{diag}(m_j^{-1} I_d)$. By the observations above,
$A_{kj}(t)$ is continuous and bounded. By standard results, the linear non-autonomous  ODE
$$
\dot T_{kj} = A_{kj}(t) T_{kj}
$$
admits a unique global solution $T_{kj}(t)$. Hence, $X_{kj}(t), Y_{kj}(t), Z_{kj}(t), W_{kj}(t)$ are unique for all $t \in \mathbb{R}$. Since this holds for any pair $k,j$, we are done.
\end{proof}

\subsection{Second-order differential equations and integral formulas}
We study the evolution of the Jacobian block matrices defined above. \\\\
Consider $X_{kj}(t).$
From the equations of motion, we know that:
\[
\dot{X}_{kj}(t) = \frac{1}{m_k} Y_{kj}(t),
\]
where \( Y_{kj}(t) := \frac{\partial p_k(t)}{\partial q_j}(P(0),Q(0)) \in \mathbb{R}^{d \times d} \).
Differentiating both sides with respect to time gives:
\[
\ddot{X}_{kj}(t) = \frac{d}{dt} \dot{X}_{kj}(t) = \frac{1}{m_k} \dot{Y}_{kj}(t).
\]
From the linearized equation for \( \dot{Y}_{kj}(t) \), we have:
\[
\dot{Y}_{kj}(t) = - \sum_{l \in \Lambda} B_{kl}(t) X_{lj}(t),
\]
where \( B_{kl}(t) \in \mathbb{R}^{d \times d} \) is defined as in Lemma \ref{Lemm: Hessian}.
Substituting into the second derivative of \( X_{kl}(t) \), we obtain the second-order evolution equation:
\[
\ddot{X}_{kj}(t) = -\frac{1}{m_k} \sum_{l \in \Lambda} B_{kl}(t) X_{lj}(t).
\]
This equation holds with the initial conditions:
\[
X_{kj}(0) = \delta_{kj} I_d, \quad \dot{X}_{kj}(0) = \frac{1}{m_k} Y_{kj}(0) = 0.
\]
Since the initial velocity is zero, we apply Duhamel's principle to solve the second-order system:
\[
X_{kj}(t) = \delta_{kj} I_d - \int_0^t (t - s) \cdot \frac{1}{m_k} \sum_{l \in \Lambda} B_{kj}(s) X_{lj}(s) \, ds.
\]
To verify this formula, we compute the time derivative:
\begin{align*}
\dot{X}_{kl}(t)
&= -\frac{d}{dt} \left( \int_0^t (t - s) \cdot \frac{1}{m_k} \sum_{j \in \Lambda} B_{kj}(s) X_{jl}(s) \, ds \right) \\
&= - \int_0^t \frac{\partial}{\partial t} (t - s) \cdot \frac{1}{m_k} \sum_{j \in \Lambda} B_{kj}(s) X_{jl}(s) \, ds \\
&= - \int_0^t \frac{1}{m_k} \sum_{j \in \Lambda} B_{kj}(s) X_{jl}(s) \, ds
\end{align*}
where we used the differentiation rule 
\[\frac{d}{dt} \int_0^tf(t,s)\, ds = f(t,t) + \int_0^t\frac{\partial f}{\partial t}(t,s) \, ds. \]
Differentiating once more gives:
\begin{align*}
\ddot{X}_{kl}(t)
&= -\frac{d}{dt} \left( \int_0^t \frac{1}{m_k} \sum_{j \in \Lambda} B_{kj}(s) X_{jl}(s) \, ds \right) \\
&= - \frac{1}{m_k} \sum_{j \in \Lambda} B_{kj}(t) X_{jl}(t),
\end{align*}
which confirms the second-order equation.
\\\\
Analogous computations show that $Y_{kj}(t)$, $Z_{kj}(t)$ and $W_{kj}(t)$ satisfy expressions of similar type. We omit the computations for brevity.
%
In summary, the full system of evolution equations can be written as:
$$
\begin{aligned}
X_{kj}(t) &= \delta_{kj}I_d -\int_0^t (t - s) \frac{1}{m_k}\sum_{l\in\Lambda} B_{kl}(s) X_{lj}(s) \, ds, \quad X_{kj}(0)=\delta_{kj}I_d,\quad \dot{X}_{kj}(0)=0;\\
Y_{kj}(t) &=  -\int_0^t \sum_{l\in\Lambda}B_{kj}(s) X_{jl}(s) \, ds, \quad Y_{kj}(0)=0,\quad  \dot{Y}_{kj}(0)=-B_{kj}(0);\\
Z_{kj}(t) &= \frac{t}{m_k}\delta_{kj}I_d - \int_0^t (t - s) \frac{1}{m_k} \sum_{l\in\Lambda}B_{kl}(s) Z_{lj}(s) \, ds, \quad Z_{kj}(0)=0,\quad \dot{Z}_{kj}(0) =\frac{1}{m_k}\delta_{kj}I_d; \\
W_{kl}(t) &= \delta_{kj}I_d - \int_0^t\sum_{l\in\Lambda} B_{kl}(s) Z_{lj}(s) \, ds, \quad W_{kj}(0)=\delta_{kj}I_d,\quad \dot{W}_{kj}(0)=0.
\end{aligned}
$$

\subsection{Estimates}
We now derive operator norm bounds for the entries of the matrix $B_{kj}(t)$, using the assumed decay of the interaction potentials, \textit{cf.} Assumption \ref{conditions}.

In particular, given (iii) the potentials $V_{kj}\colon \mathbb{R}^d \to \mathbb{R}$ satisfy the following exponential decay condition on their second derivatives: there exist constants $C_V$  and $C_{kj}$ for all $j, k \in \Gamma$, such that 
$$
\|\Hess( V_{kj} )\|_{\op,\infty} \leq dC_{kj} C_V^2,
$$
where $C_{kj} \leq \|\Psi\| F(d(k,j))$, with $F\colon \mathbb{R}_{\geq 0} \to (0,1]$ a suitably decaying function, see Section \ref{conditions} for the precise definition. Therefore,
$$
\|\Hess( V_{kj} )\|_{\op,\infty} \leq dC_V^2 \|\Psi\| F(d(k,j)).
$$
Recall from Lemma \ref{Lemm: Hessian} that $B_{kj}(t)$ is defined via the second derivatives of the interaction potentials
$$
B_{kj}(t) = \begin{cases}
\displaystyle \nu_{j}+\sum_{l \in \Lambda} \Hess( V_{lj})(Q_j(t) - Q_l(t)), & j = k, \\[6pt]
- \Hess( V_{kj})(Q_k(t) - Q_j(t)), & j \neq k.
\end{cases}
$$
We thus obtain
$$
\|B_{kj}(t)\|_{\infty} \leq
\begin{cases}
\displaystyle \nu_j+\sum_{l \in \Lambda} \left\| \Hess( V_{jl})(Q_j(t) - Q_l(t)) \right\|_{\op} \leq \nu_j+\sum_{l \in \Lambda} dC_V^2 \|\Psi\| F(d(j,l)), & j = k, \\[10pt]
\displaystyle dC_V^2 \|\Psi\| F(d(k,j)), & j \neq k.
\end{cases}
$$
\subsubsection{Dyson expansions}
We derive Dyson expansions for each Jacobian matrix. We give the proofs for 
$X_{kj}(t)$ and $Y_{kj}(t)$ explicitly, and leave the analogous computations for 
$Z_{kj}(t)$ and $W_{kj}(t)$ to the reader. Given the similarities between $X_{kj}(t)$ and $Z_{kj}(t)$, as well as between $Y_{kj}(t)$ and $W_{kj}(t)$, their computations are nearly identical.

\begin{proposition}\label{prop: dysons}
Let $t\in\mathbb{R}$. For all $j\neq k$ it holds
\begin{align}
\begin{cases}
\displaystyle \|X_{jk}(t)\|_{\op} \leq F(d(j,k)) (\cosh(\sqrt{C_0} t)-1), \\[8pt]
\displaystyle \|Y_{jk}(t)\|_{\op} \leq F(d(j,k)) \sqrt{C_0}\sinh(\sqrt{C_0} |t|), \\[8pt]
\displaystyle \|Z_{jk}(t)\|_{\op} \leq   F(d(j,k)) \frac{\sinh(\sqrt{C_0} |t|)}{\sqrt{C_0}}, \\[8pt]
\displaystyle \|W_{jk}(t)\|_{\op} \leq  F(d(j,k)) (\cosh(\sqrt{C_0} t)-1)
\end{cases}
\end{align}
In particular, because the right-hand side does not depend on the choice $P(0),Q(0)$ in the definition of $X_{j,k}(t)$, $Y_{j,k}(t)$, $Z_{j,k}(t)$ and $W_{j,k}(t)$, we have
\begin{align}\label{eq:cases}
\begin{cases}
\displaystyle\left\| \frac{\partial q_j(t)}{\partial q_k}\right\|_{\op,\infty} \leq F(d(j,k)) (\cosh(\sqrt{C_0} t)-1), \\[8pt]
\displaystyle \left\| \frac{\partial p_j(t)}{\partial q_k}\right\|_{\op,\infty} \leq F(d(j,k)) \sqrt{C_0}\sinh(\sqrt{C_0} |t|), \\[8pt]
\displaystyle \left\| \frac{\partial q_j(t)}{\partial p_k}\right\|_{\op,\infty} \leq   F(d(j,k)) \frac{\sinh(\sqrt{C_0} |t|)}{\sqrt{C_0}}, \\[8pt]
\displaystyle \left\| \frac{\partial p_j(t)}{\partial p_k}\right\|_{\op,\infty}\leq  F(d(j,k)) (\cosh(\sqrt{C_0} t)-1)
\end{cases}
\end{align}
\end{proposition}
These decay estimates show that each matrix entry inherits the exponential decay in space via $F(d(j,k))$, with explicit controlled time growth.\\
To prove the proposition, we first deduce the Dyson expansion for \(X(t)\), stated in the following lemma.
\begin{lemma}\label{lemma: X}
For $n \geq 0$ and all $j, k \in \Lambda$, let
$$
X_{kj}^{(n+1)}(t):= -\int_0^t (t-s) \frac{1}{m_k}\sum_{l\in\Lambda}B_{kl}(s) X_{lj}^{(n)}(s) \, ds,
$$
and  $X_{kj}^{(0)}(t)=\delta_{kj}I_d$. Then, whenever $j\neq k$,
$$
\| X_{kj}^{(n)}(t) \|_{\op} \leq \frac{(C_0 t^2)^n}{(2n)!} F(d(k,j)),
$$
where
$$
C_0 := \|m^{-1}\|_\infty \max{\{(\|\nu\|_\infty +dC_V^2\|\Psi\|\|F\|),dC_V^2\|\Psi\|C_F,1\}}.
$$
Here, $\|m^{-1}\|_\infty=\sup_{j}|m_j^{-1}|$ and $\|\nu\|_\infty=\sup_j|\nu_j|$, and the constants $C_F$ and $\|F\|$ are the ones defined in Assumption \ref{conditions}.
In particular, whenever $k\neq j$
$$\sum_{n=0}^\infty\| X_{kj}^{(n)}(t)\|_{\op}\leq F(d(k,j)) (\cosh(\sqrt{C_0} t)-1),$$
and $X_{kj}(t)=\sum_{n=0}^\infty X_{kj}^{(n)}(t)$. 
\end{lemma}
\begin{proof}
The first assertion is proved by induction.
\\\\
\noindent \textbf{Base case $n=0$:}\\
By definition,
$$
X_{kj}^{(0)}(t) = \delta_{kj} I_d.
$$
and $F(d(k,j))\geq 0$.
In particular, when $k\neq j$, we obtain $0=\| X_{kj}^{(0)}(t)\|\leq F(d(k,j))$.
\\\\
\noindent \textbf{Inductive step:}\\ Assume the claim holds for some $n \geq 0$. Then,
$$
\| X_{kj}^{(n+1)}(t) \|_{\op} \leq \int_0^{|t|} (|t| - s) \| m^{-1} \|_\infty \sum_{l \in \Lambda} \| B_{kl}(s) \|_{\op} \cdot \| X_{lj}^{(n)}(s) \|_{\op} ds.
$$
Split the sum over $l$ into diagonal and off-diagonal parts:
$$
\sum_{l \in \Lambda} \| B_{kl}(s) \|_{\op} \| X_{lj}^{(n)}(s) \|_{\op} = \| B_{kk}(s) \|_{\op} \| X_{kj}^{(n)}(s) \|_{\op} + \sum_{l \neq k} \| B_{kl}(s) \|_{\op} \| X_{lj}^{(n)}(s) \|_{\op}.
$$
Recall that the matrix elements of $B_{kk}(t)$ satisfy
$$
\| B_{kk}(t) \|_{\op} \leq \|\nu\|_\infty + dC_V^2\|\Psi\|\sum_{m\in\Lambda}F(d(k,m)),
$$
and for $k\neq l$
$$
\| B_{kl}(t) \|_{\op} \leq dC_V^2\|\Psi\|F(d(k,l)),
$$
where $F$ is the decay function (with $F(0) =1$).
We can absorb the diagonal contribution into the decay estimates and obtain
\begin{align*}
    &\sum_{l \in \Lambda} \| B_{kl}(s) \|_{\op} \| X_{lj}^{(n)}(s) \|_{\op} \leq\\&
     \left(\|\nu\|_\infty+dC_V^2 \|\Psi\| \sum_{m \in\Lambda}  F(d(k,m)) \right) \frac{(C_0 s^2)^n}{(2n)!} F(d(k,j)) + \sum_{l \neq k} dC_V^2 \|\Psi\| F(d(k,l)) \frac{(C_0 s^2)^n}{(2n)!} F(d(l,j)),
\end{align*}
where we have used the induction hypothesis.
Applying the exponential decay properties and the convolution bound, there exist constants $C_F, \|F\| > 0$, i.e. such that
$$
\sup_k\sum_{m \in\Lambda} F(d(k,m)) \leq \|F\|; \quad \text{and} \quad \sum_{l \in \Lambda} F(d(k,l)) F(d(l,j)) \leq C_F F(d(k,j)).
$$
Combining these estimates yields,
\begin{align*}
&\sum_{l \in \Lambda} \| B_{kl}(s) \|_{\op} \| X_{lj}^{(n)}(s) \|_{\op}\leq \\& (\|\nu\|_\infty +dC_V^2\|\Psi\|\|F\|)\frac{(C_0 s^2)^n}{(2n)!} F(d(k,j))+ dC_V^2 \|\Psi\| C_F\frac{(C_0 s^2)^n}{(2n)!} F(d(k,j))).
\end{align*}
Since
$$
C_0= \|m^{-1}\|_\infty \max{\{(\|\nu\|_\infty +dC_V^2\|\Psi\|\|F\|),dC_V^2\|\Psi\|C_F,1\}},
$$
it follows that
$$
\| m^{-1} \|_\infty\sum_{l \in \Lambda} \| B_{kl}(s) \|_{\op} \| X_{lj}^{(n)}(s) \|_{\op}\leq C_0 F(d(k,j))\frac{(C_0s^2)^n}{(2n)!}.$$
Hence,
$$
\| X_{kj}^{(n+1)}(t) \|_{\op} \leq  F(d(k,j)) \frac{C_0^{n+1}}{(2n)!} \int_0^{|t|} (|t| - s) s^{2n} ds.
$$
We then evaluate the integral using the substitution $u = s/|t|$ and the properties of a Beta integral,
$$
\int_0^{|t|} (|t| - s) s^{2n} ds = |t|^{2n+2} \int_0^1 (1 - u) u^{2n} du = \frac{|t|^{2n+2}}{(2n+1)(2n+2)}.
$$
Therefore,
$$
\| X_{kj}^{(n+1)}(t) \|_{\op} \leq  F(d(k,j)) \frac{C_0^{n+1} |t|^{2n+2}}{(2n)! (2n+1)(2n+2)}
$$
which is equal to
$$
\| X_{kj}^{(n+1)}(t) \|_{\op} \leq F(d(k,j)) \frac{(C_0 t^2)^{n+1}}{(2(n+1))!},
$$
By the principle of induction,
$$
\| X_{kj}^{(n)}(t) \|_{\op} \leq F(d(k,j)) \frac{(C_0 t^2)^{n}}{(2n)!},
$$
for all $n\geq 0$. Finally, for $k \neq j$ we sum over $n \geq 1$:
\[
 \sum_{n=0}^\infty \left\| X_{kj}^{(n)}(t) \right\|_{\op} =  \sum_{n=1}^\infty \left\| X_{kj}^{(n)}(t) \right\|_{\op} \leq F(d(k,j)) \sum_{n=1}^\infty \frac{(C_0 t^2)^n}{(2n)!} = F(d(k,j)) (\cosh(\sqrt{C_0} t)-1).
\]

This proves that the  Dyson series defined by
$$
\sum_{n=0}^\infty X_{kj}^{(n)}(t), \quad \text{where} \quad X_{kj}^{(0)}(t) = \delta_{kj}I_d,
$$
converges absolutely. A direct computation shows that this series satisfies the second order differential equation for $X_{kj}$. On account of Lemma \ref{lemma: uniqueness}, we conclude that
$$
X_{kj}(t)=\sum_{n=0}^\infty X_{kj}^{(n)}(t), \quad \text{where} \quad X_{kj}^{(0)}(t) = \delta_{kj}I_d.
$$
\end{proof}
We now derive the Dyson expansion for $Y_{kj}(t)$.
\begin{lemma}
For $n\geq 1$ and all $k,j\in\Lambda$, define
$$
Y_{kj}^{(n)}(t) = -\int_0^t \sum_{l\in\Lambda}B_{kl}(s) X_{lj}^{(n-1)}(s) \, ds, 
$$
and $Y_{kl}^{(0)}(t)=0$. If $j\neq k$, it holds
$$
\|Y^{(n)}_{kj}(t)\|_{\op} \leq \frac{C_0^{n} |t|^{2n-1}}{(2n-1)!} F(d(k,j)),
$$
and $\|Y^{(0)}_{kj}(t)\|_{\op} =0$. Here, the constant $C_0$ is the one defined by Lemma \ref{lemma: X}.
In particular, if $k\neq j$
$$
\sum_{n=0}^\infty \|Y_{kj}^{(n)}(t)\|_{\op} \leq F(d(k,j)) \sqrt{C_0}\sinh(\sqrt{C_0} |t|)
$$
and 
$$
Y_{kj}(t)=\sum_{n=0}^\infty Y_{kj}^{(n)}.(t)
$$
\end{lemma}
\begin{proof}
We proceed by induction on $n$.\\\\
\textbf{Base case ($n=1$):}\\ From the definition,
\[
Y_{kj}^{(1)}(t) = -\int_0^t \sum_{l \in \Lambda} B_{kl}(s) X_{lj}^{(0)}(s) \, ds.
\]
Since $X_{lj}^{(0)}(s)=\delta_{lj}I_d$, we have
\[
\|Y^{(1)}_{kj}(t)\|_{\op} \leq \int_0^{|t|}  \|B_{kj}(s)\|_{\op} \, ds
\leq C_0 F(d(k,j)) \int_0^{|t|}ds=|t|C_0F(d(k,j)),
\]
where the final last inequality follows from the estimates on $B_{kj}$ for $k\neq j$.
\textbf{Inductive step:}\\
Assume that
\[
\|Y^{(n)}_{kj}(t)\|_{\op} \leq \frac{C_0^n t^{2n - 1}}{(2n - 1)!} F(d(k,j)),
\]
holds for some $n \geq 1$. We want to prove it for $n+1$. Using the recursive definition:
\[
Y^{(n+1)}_{kj}(t) = -\int_0^t \sum_{l \in \Lambda} B_{kl}(s) X^{(n)}_{lj}(s) \, ds,
\]
we estimate the matrix elements as follows:
$$
\|Y^{(n+1)}_{kj}(t)\|_{\op} \leq \int_0^{|t|} \sum_{l \in \Lambda} \|B_{kl}(s)\|_{\op} \cdot \|X^{(n)}_{lj}(s)\|_{\op} \, ds.
$$
Using the bounds on $ X^{(n)}_{lj}(s)$ and $B_{kl}(s)$ (cf. Lemma \ref{lemma: X}), we obtain
$$
\sum_{l \in \Lambda} \|B_{kl}(s)\|_{\op} \cdot \|X^{(n)}_{lk}(s)\|_{\op} \leq  \frac{C_0^{n+1} s^{2n}}{(2n)!} F(d(k,j)).
$$
Then,
$$
\|Y^{(n+1)}_{kj}(t)\|_{\op} \leq  F(d(k,j)) \cdot \frac{C_0^{n+1}}{(2n)!} \int_0^{|t|} s^{2n} \, ds = F(d(k,j)) \cdot \frac{C_0^{n+1} |t|^{2n+1}}{(2n)!(2n+1)}.
$$
As a result,
$$\|Y^{(n+1)}_{kj}(t)\|_{\op}\leq  F(d(k,j)) \cdot \frac{C_0^{n+1} |t|^{2n+1}}{(2n+1)!}.$$
This completes the induction step. Thus, for all $n\geq 1$ and $k \ne j$,
\[
\|Y^{(n)}_{kj}(t)\|_{\op} \leq \frac{C_0^n |t|^{2n-1}}{(2n-1)!} F(d(k,j)).
\]
Summing this series for $j\neq k$ gives
$$
\sum_{n=0}^\infty\| Y_{jk}^{(n)}(t)\|_{\op}\leq\sum_{n=1}^\infty\frac{C_0^{n} |t|^{2n-1}}{(2n-1)!} F(d(k,j))=F(d(k,j))\sqrt{C_0}\sinh{\sqrt{C_0}|t|}.
$$
This shows that the Dyson series
\[
Y_{kj}(t) = \sum_{n=0}^\infty Y_{kj}^{(n)}(t), \quad \text{with } Y_{kj}^{(0)}(t) = 0,
\]
converges absolutely and satisfies the relevant integral equation. By uniqueness (cf.\ Lemma~\ref{lemma: uniqueness}), this series is the solution.
\end{proof}

\begin{theorem}\label{thm: main 1}
Assume all the conditions in Section \ref{sct:subsection The phase space} are satisfied for the discrete set $\Gamma$.
Let $H_\Lambda$ be defined by \eqref{Hamiltoniannew2}.
For $X,Y\Subset\Gamma$ be  such that they are spatially separated, i.e. $\text{dist}(X,Y)>0$.  Take any finite set $\Lambda\Subset\Gamma$ with  $X\cup Y\subset\Lambda$.  Let $f\in C_b^1(\Omega_{X})$ and
$g\in C_b^1(\Omega_{Y})$. Then, there exists a positive constant $C_0$ independent of $f$ and $g$, such that for all $t\in\mathbb{R}$ 
\begin{align}
\|\{\alpha_\Lambda^t(f),g\}\|_\infty\leq 4\|f\|_{C^1} \|g\|_{C^1}\sqrt{C_0}\sinh{\sqrt{C_0}|t|}\text{D}(X,Y),
\end{align}
where
$$
D(X,Y):=\sum_{x \in X} \sum_{y \in Y} F(d(x,y)).
$$
In particular, for all $t\in\mathbb{R}$
\begin{align}
\|\{\alpha_\Lambda^t(f),g\}\|_\infty\leq 4\|f\|_{C^1} \|g\|_{C^1}\sqrt{C_0}e(^{\sqrt{C_0}|t|}-1)\text{D}(X,Y),
\end{align}
\end{theorem}
\begin{proof}
Let us first assume $t>0$.
We already know from the computations in Section \ref{sec:estimates}
\begin{align*}
&\|\{\alpha_t(f), g\}\|_\infty\leq 4\|f\|_{C^1}\|g\|_{C^1} \sup_{\substack{j \in X \\ k \in Y}} \left\{ \left\| \frac{\partial q_j(t)}{\partial q_k}\right\|_{\op,\infty}, \left\| \frac{\partial p_j(t)}{\partial q_k}\right\|_{\op,\infty}, \left\| \frac{\partial q_j(t)}{\partial p_k}\right\|_{\op,\infty}, \left\| \frac{\partial p_j(t)}{\partial p_k}\right\|_{\op,\infty} \right\},
\end{align*}
Moreover, for $t\in\mathbb{R}$
\begin{align*}
&\max{\bigg{\{} \left\| \frac{\partial q_j(t)}{\partial q_k}\right\|_{\op,\infty}, \left\| \frac{\partial p_j(t)}{\partial q_k}\right\|_{\op,\infty}, \left\| \frac{\partial q_j(t)}{\partial p_k}\right\|_{\op,\infty}, \left\| \frac{\partial p_j(t)}{\partial p_k}\right\|_{\op,\infty}\bigg{\}}} \\&\leq
\sqrt{C_0} F(d(j,k))\sinh(\sqrt{C_0} |t|),
\end{align*}
where
$$
C_0 := \|m^{-1}\|_\infty \max{\{(\|\nu\|_\infty +dC_V^2\|\Psi\|\|F\|),dC_V^2\|\Psi\|C_F,1\}}.
$$
Hence,
\begin{align*}
&\|\{\alpha_t(f), g\}\|_\infty\leq 4\|f\|_{C^1} \|g\|_{C^1}\sqrt{C_0}\sinh(\sqrt{C_0} t)D(X,Y).
\end{align*}
Note also that the case $t=0$ implies that
$$
\|\{\alpha_0(f), g\}\|_\infty=0,
$$
as it must be, since $f,g$ have disjoint support.
In particular, for all $t\in\mathbb{R}$
\begin{align}
\|\{\alpha_\Lambda^t(f),g\}\|_\infty\leq 4\|f\|_{C^1} \|g\|_{C^1}\sqrt{C_0}(e^{\sqrt{C_0}|t|}-1)\text{D}(X,Y),
\end{align}
which follows from the estimate
\begin{align*}
\sinh(\sqrt{C_0}|t|) 
=\frac{1}{2} \left(e^{\sqrt{C_0} |t|}-e^{-\sqrt{C_0} t}\right) \leq \frac{1}{2}(e^{\sqrt{C_0} |t|}-1).
\end{align*}
\end{proof}
Note that, if $\Gamma $ is equipped with such a function $F$, then for every $ \mu > 0 $, the function
\begin{equation}
F_\mu(r) = e^{-\mu r} F(r)
\end{equation}
also satisfies the convolution condition with the same constant $C_F$ (as follows from the triangle inequality). Moreover $\|F_\mu\|\leq \|F\|$. Setting then
\begin{equation}
\|\Psi\|_\mu := \sup_{k,j \in \Gamma} \frac{C_{kj}}{F_\mu(d(k,j))}, \quad \mu>0;
\end{equation}
it is easy to see that for all $t \in\mathbb{R}$,
\begin{align*}
&\|\{\alpha_\Lambda^t(f),g\}\|_\infty\leq
 4\|f\|_{C^1} \|g\|_{C^1}\sqrt{C_\mu}\sinh{(\sqrt{C_\mu} |t|)}\sum_{x \in X} \sum_{y \in Y} e^{-\mu \text{dist}(x,y)}F(d(x,y))\\&\leq
 4\|f\|_{C^1} \|g\|_{C^1}\sqrt{C_\mu}\sinh{(\sqrt{C_\mu}|t|)}e^{-\mu \text{dist}(X,Y)}\min{\{|X|,|Y|\}}\max_{y\in\Gamma}\sum_{x \in \Gamma}  F(d(x,y))\\&=
 4\|f\|_{C^1} \|g\|_{C^1}\sqrt{C_\mu}\min{\{|X|,|Y|\}}\|F\|\sinh{(\sqrt{C_\mu}|t|)}e^{-\mu \text{dist}(X,Y)},
\end{align*}
where $C_\mu$ is given by $C_0$ with $\|\Psi\|$ replaced by $\|\Psi_\mu\|$.
In particular, 
\begin{align*}
\sinh{(\sqrt{C_\mu}|t|)} e^{-\mu \operatorname{dist}(X,Y)} \leq \frac{1}{2} e^{\sqrt{C_\mu}|t| - \mu \operatorname{dist}(X,Y)}.
\end{align*}
Hence, one obtains exponential decay in \( \operatorname{dist}(X,Y) \) whenever \(  |t| < \frac{\mu}{\sqrt{C_\mu}} \operatorname{dist}(X,Y) \). Notice that this holds for all fixed \( \operatorname{dist}(X,Y) \) if \( |t| \) is sufficiently small. 
Therefore, we have obtained the following corollary.

\begin{corollary}
Assume the conditions of Theorem \ref{thm: main 1}. For all $\mu>0$ and $t\in\mathbb{R}$, it holds
\begin{align*}
\|\{\alpha_\Lambda^t(f),g\}\|_\infty\leq
 C\|f\|_{C^1} \|g\|_{C^1}e^{- \mu( \operatorname{dist}(X,Y)-\frac{\sqrt{C_\mu} |t|}{\mu})},
\end{align*}
for come constant $C>0$ depending on the Hamiltonian, the supports $X$ and $Y$ of the observables $f$ and $g$, respectively,  on the function $F$, and on $\mu$.
\end{corollary}

\section{Existence of the global dynamics}\label{Section: existence}
The first step in defining  global dynamical system, is the construction of a $C^*-$algebra. Since we are dealing with infinite particles lattices systems for which the single-site phase space is non-compact and unbounded, one has to be careful.
A promising and convenient framework is provided by the {\em commutative resolvent algebra} \cite{vanNuland_2019}, a recent development that has been shown to arise as the strict classical limit of the non-commutative resolvent algebra originally introduced by Buchholz and Grundling. This algebra is specifically tailored to accommodate the difficulties presented by unbounded configuration and momentum variables, making it a natural setting for describing observables of infinite classical systems \cite{Nuland_Ven_2023}.
Importantly,  if $X$ is an infinite dimensional normed vector space, the commutative resolvent algebra $\CR(X)$ can be constructed as the inductive limit of algebras $C_0(V)$ over finite-dimensional subspaces $V \subseteq X$, capturing the quasi-local nature of observables in a natural and well-behaved manner \cite{Buchholz_Grundling_2008,vanNuland_2019}. This inductive limit structure is crucial for studying the thermodynamic limit: infinite systems emerge as limits of increasingly large finite subsystems. This approach facilitates a rigorous algebraic description of physically significant phenomena such as spontaneous symmetry breaking and phase transitions, which are inherently infinite-volume effects, yet arise from finite approximations \cite{Moretti_vandeVen_2021,vandeVen_2022,vandeVen_2023}.
\subsection{The commutative resolvent algebra}\label{Commutative resolvent algebra}

The commutative resolvent algebra is defined as follows. Let $X$ be a (possibly infinite-dimensional) real-linear inner product space. The {\em commutative resolvent algebra} of $X$, denoted $\CR(X)$, is the C*-subalgebra of the algebra of bounded operators $C_\b(X)$ generated by resolvent functions on $X$, i.e.,  functions of the form
\begin{align*}
	h_x^\lambda(y):=1/(i\lambda- x\cdot y),
\end{align*}
for $x\in X$, $\lambda\in\mathbb{R}\setminus\{0\}$. Here, the inner product $\cdot$ is the one standard one induced by the complex structure compatible with the symplectic form. 
The inner product on $X$ yields a norm $||\cdot||$ and a topology with respect to which $h_x^\lambda:X\to\mathbb C$ is a continuous function.
\\\\
We now consider $\CR(\Omega)$, with $\Omega=\ell_c(\Gamma,\R^{2r})\subseteq\ell^2(\Gamma,\R^{2r})$ defined in Section \ref{sct:subsection The phase space}. As a result, $\CR(\Omega)$ is the inductive limit of the net of all $\CR(V)$, where $V\subset \Omega$  ranges over all finite-dimensional subspaces of $\Omega$, and the connecting maps defining this limit are the pull-backs of the projection maps $W\twoheadrightarrow V$ for $V\subset W$. This remains true if one restricts the net to any cofinal class of finite-dimensional subspaces. It follows that 
$$\CR(\Omega)=\varinjlim \CR(\Omega_\Lambda),$$ where the inductive limit is taken with respect to the pull-backs of ${\pi_\Lambda}|_{\Lambda'}:\Omega_{\Lambda'}\to\Omega_{\Lambda}$ $~(\Lambda\subseteq\Lambda'\Subset\Gamma)$, where $\pi_\Lambda:\Omega\to\Omega_{\Lambda}$ is the orthogonal projection onto $\Omega_\Lambda$.
Hence, $\CR(\Omega)$ is an algebra of ``quasi-local'' observables, containing a dense subalgebra of ``local'' observables, i.e. functions that only depend on finitely many particles.
\\\\
We  define the subspace $\SR(\Omega)\subset \CR(\Omega)$ as the span of so-called levees $g\circ \pi$ for which $g$ is Schwartz, namely
\begin{align}\label{eq: levees}
	\SR(\Omega):=\text{span}\{ g\circ \pi ~|~\pi\text{ fin. dim. projection on $\Omega$, }\ g\in \S(\ran(\pi)) \},
\end{align}
where $\S(\text{ran}(\pi))$ denotes the Schwartz space on $\text{ran}(\pi)$. More generally, a ``levee'' is a function $f = g \circ \pi\in \CR(\Omega)$ for a finite dimensional projection $\pi$ and a function $g \in C_0(\text{ran}(\pi))$. By \cite[Prop. 2.4]{vanNuland_2019} the set $\SR(\Omega)$ is a dense *-subalgebra of $\CR(\Omega)$.

We can put a Poisson bracket on $\SR(\Omega)$ by use of the canonical Poisson bracket on $C^\infty(\Omega_\Lambda)\cong C^\infty(\R^{2|\Lambda|r})$, as follows. 
For any two functions $f_1,f_2\in \SR(\Omega)$ we can choose $\Lambda\Subset\Gamma$ large enough such that $f_1=g_1\circ \pi_{\Lambda}$ and $f_2=g_2\circ \pi_\Lambda$ for functions $g_1,g_2\in \SR(\Omega_\Lambda)\subseteq C^\infty(\Omega_\Lambda)$, and where $\pi_\Lambda:\Omega\to\Omega_\Lambda$ denotes the orthogonal projection.
We define
\begin{align*}
    \{g_1\circ \pi_{\Lambda},g_2\circ \pi_{\Lambda}\}:=\{g_1,g_2\}_\Lambda\circ \pi_{\Lambda},
\end{align*}
where $\{g_1,g_2\}_\Lambda$ is defined by
\begin{align*}
	\{g_1,g_2\}_{\Lambda}(p,q):=\sum_{j\in\Lambda}\sum_{i=1}^d \bigg(\frac{\partial g_1(p,q)}{\partial q_{j,i}}\frac{\partial g_2(p,q)}{\partial p_{j,i}}-\frac{\partial g_1(p,q)}{\partial p_{j,i}}\frac{\partial g_2(p,q)}{\partial q_{j,i}}\bigg),
\end{align*}
for all $(p,q)\in\Omega_\Lambda$.  One can prove that $\{g_1\circ \pi_{\Lambda},g_2\circ \pi_{\Lambda}\}$ does not depend on $\Lambda$ and lands in $\SR(\Omega)$. Note further that this bracket coincides with the one defined in Section \eqref{sec:estimates}.

This algebraic formalism of resolvent algebras allows one to derive a global $C^*$-dynamical system. This has been particularly proved in \cite[Thm. 14]{Nuland_Ven_2023}. In this paper, we provide an alternative proof making use of the Lieb-Robinson bounds obtained in the previous section. The main result is the following theorem.

\begin{theorem}[Existence of Infinite-Volume Dynamics]\label{Thm:existence}
Assume $\Gamma$ is equipped with a function $F$ satisfying the conditions in Section \ref{Set-up}.
Then, the local automorphism $\alpha_t^\Lambda$ converges to a globally defined automorphism of the commutative resolvent algebra $\CR(\Omega)$, as $\Lambda\to\infty$. Moreover,  convergence is uniform for $t$ in the compact intervals $[-T, T]$. 
\end{theorem}
\begin{proof}
Let us first take $t>0$.
Fix a local observable 
\[
f \in \SR(\Omega_X) \subset \CR(\Omega_X), 
\]
supported in a finite set \(X \subset \Gamma\). Let \(\Lambda_1 \subset \Lambda_2\) be finite subsets of \(\Gamma\) such that \(X \subset \Lambda_1\).
For a finite volume \(\Lambda \subset \Gamma\), write
\begin{equation}
H_\Lambda = H_\Lambda^{har} + H_\Lambda^{int}, \quad
H_\Lambda^{ har} = \sum_{k \in \Lambda} \bigg{(}\frac{\|p_k\|^2}{2m_k} + \frac{\nu_k \|q_k\|^2}{2}\bigg{)}, \quad
H_\Lambda^{int} = \sum_{Z \subset \Lambda} \Psi(Z),
\end{equation}
where
$$\Psi(Z)(p,q):=\begin{cases}
   & \frac{1}{2}V_{kl}(q_k-q_l),\quad \quad \text{ if}\quad Z=\{k, l\};\\&
    0, \quad \quad \text{else}\
    \end{cases}
$$
We denote by $\Phi_t^{0,\Lambda}\) the flow of the harmonic oscillator, the time evolution by 
$$
\alpha_t^{0,\Lambda}(f) := f \circ \Phi_t^{0,\Lambda},\quad f\in \CR(\Omega_{\Lambda}).
$$
Since $\Phi_t^{0,\Lambda}$ is only a rotation, the pullback leaves the commutative resolvent algebra $\CR(\Omega_{\Lambda})$ invariant, i.e. 
$$
f\in \CR(\Omega_{\Lambda})\implies \alpha_t^{0,\Lambda}(f)\in  \CR(\Omega_{\Lambda}).
$$
We also define the interaction-picture observable
\begin{equation*}
\gamma_t^\Lambda(f) :=f\circ\Phi_{-t}^{0,\Lambda}\circ\Phi_t^\Lambda,
\end{equation*}
where $\Phi_t^\Lambda$ denotes the flow of the full Hamiltonian $H_\Lambda$.
It satisfies the time-dependent equation
\begin{equation}
\frac{d}{dt} \gamma_t^\Lambda(f) = \big\{ \tilde H_\Lambda^{ int}(t), \gamma_t^\Lambda(f) \big\}, \quad \gamma_0^\Lambda(f) = f,
\end{equation}
with the interaction-picture Hamiltonian
\begin{equation}
\tilde H_\Lambda^{ int}(t) := \alpha_{-t}^{0,\Lambda}(H_\Lambda^{ int}) = \sum_{Z \subset \Lambda} \Psi(Z) \circ \Phi_{-t}^{0,\Lambda}.
\end{equation}
Let us denote by $\mathcal{B}$ the set of all  $Z\subset \Lambda_2$ for which  $Z \cap (\Lambda_2 \setminus \Lambda_1) \neq \emptyset$. Then, defining
\[
F(s) := \gamma_s^{\Lambda_2} \big( \gamma_{t-s}^{\Lambda_1}(f) \big), \quad s \in [0,t],
\]
it follows that
\begin{align}\label{fund}
 \gamma_t^{\Lambda_2}(f) - \gamma_t^{\Lambda_1}(f)
&= F(t)-F(0)=\int_0^t \frac{d}{ds} F(s) \, ds\nonumber \\&=\int_0^t ds \, \gamma_s^{\Lambda_2} \Big( \big\{ \tilde H_{\Lambda_2}^{int}(s) - \tilde H_{\Lambda_1}^{ int}(s), \gamma_{t-s}^{\Lambda_1}(f) \big\} \Big),
\end{align}
with
\begin{equation*}
\tilde H_{\Lambda_2}^{ int}(s) - \tilde H_{\Lambda_1}^{int}(s) = \sum_{Z \in\mathcal{B}} \alpha_{-s}^{0,\Lambda_2}(\Psi(Z)).
\end{equation*}
We are going to use  the classical Lieb-Robinson bound for bounded interactions to estimate the norm difference
\begin{align}\label{normdiff}
    \| \gamma_t^{\Lambda_2}(f) - \gamma_t^{\Lambda_1}(f)\|_\infty.
\end{align}
To this end, note that $\gamma_{t-s}^{\Lambda_1}(f)$ is still defined in $\CR(\Omega_{\Lambda_1})$, since the local dynamics $\alpha_t^\Lambda$ leaves the algebra $\CR(\Omega_\Lambda)$ invariant for each $t\in\mathbb{R}$ and $\Lambda\Subset\Gamma$, and the part coming from the harmonic oscillator only rotates \cite[Thm. 13]{Nuland_Ven_2023}.
Let us now estimate \eqref{normdiff}. Using the identity obtained in \eqref{fund}, it holds
\begin{align*}
    \|\gamma_t^{\Lambda_2}(f) - \gamma_t^{\Lambda_1}(f)\|_\infty&\leq \sum_{\{k,j\} \in \mathcal{B}} \int_0^t ds \| \{\alpha_{-s}^{0,\Lambda_2}(\Psi(\{k,j\})),\gamma_{t-s}^{\Lambda_1}(f) \}  \|_\infty\\&=
    \sum_{\{k,j\} \in \mathcal{B}} \int_0^t d\tau \| \{\alpha_{-(t-\tau)}^{0,\Lambda_2}(\Psi(\{k,j\})),\gamma_{\tau}^{\Lambda_1}(f) \}  \|_\infty.
\end{align*}
We further estimate the integrand
\begin{align*}
&\left\{ \alpha_{-(t-\tau)}^{0,\Lambda_2}(\Psi(\{k,j\})), \gamma_{\tau}^{\Lambda_1}(f) \right\}(p,q)\\& = \bigg{(}\frac{1}{2} 
 \frac{\partial (V_{kj}\circ \Phi_{-(t-\tau)}^{0,\Lambda_2})}{\partial q_{k}}(p,q)\bigg{)}^T\frac{\partial \gamma_{\tau}^{\Lambda_1}(f)}{\partial p_{k}}(p,q)\\&+
 \bigg{(}\frac{1}{2}
 \frac{\partial (V_{kj}\circ \Phi_{-(t-\tau)}^{0,\Lambda_2})}{\partial q_{j}}(p,q) \bigg{)}^T\frac{\partial \gamma_{\tau}^{\Lambda_1}(f)}{\partial p_{j}}(p,q).
\end{align*}
As proved in \cite[Lemma 15]{Nuland_Ven_2023}, for all $t\in\mathbb{R}$ it holds
$$
\bigg{\|}\frac{\partial (V_{kj}\circ \Phi_{t}^{0,\Lambda_2})}{\partial q_{k}} \bigg{\|}_{\op,\infty}\leq dC_{kl}C_V\leq \|\Psi\| dC_V F(d(k,j)),
$$
where the last inequality following from Assumption \ref{ass:psi}.
Let us  write 
\[
\tilde \Phi_\tau(p,q) := \Phi_{-\tau}^{0,\Lambda_1} \circ \Phi_\tau^{\Lambda_1}(p,q) =  \Phi_{-\tau}^{0,\Lambda_1} (p(\tau),q(\tau))=(\tilde p (\tau),\tilde q (\tau)).
\]
Then, by the chain rule,
$$
\frac{\partial \gamma_\tau^{\Lambda_1}(f)}{\partial p_{k}}(p,q) = \sum_{j \in X} \left[\bigg{(}
\frac{\partial f}{\partial q_j} (\tilde \Phi_\tau(p,q))\bigg{)}^T \cdot \frac{\partial \tilde q_j(\tau)}{\partial p_{k}}(p,q) + \bigg{(}\frac{\partial f}{\partial p_j}(\tilde \Phi_\tau(p,q)) \bigg{)}^T\cdot \frac{\partial \tilde p_j(\tau)}{\partial p_{k}}(p,q)
\right],
$$
where
$$
\tilde q_j(\tau)(p,q)=  (\Phi_{-\tau}^{0,\Lambda_1}(p(\tau),q(\tau)))_j^q.
$$
Since  $f\in \mathcal{S}_R(\Omega_X)$ each of its (first) derivatives are uniform bounded, i.e.  $\|f\|_{C^1}<\infty$, it again follows that by the chain rule that
\begin{equation}
\frac{\partial \tilde q_j(\tau)}{\partial p_k} (p,q)
= \sum_{\ell \in X} \left(
\frac{\partial (\Phi_{-\tau}^{0,\Lambda_1})_j^q}{\partial q_\ell}(p(\tau),q(\tau)) \frac{\partial q_\ell(\tau)}{\partial p_k}(p,q)
+ \frac{\partial  (\Phi_{-\tau}^{0,\Lambda_1})_j^q}{\partial p_\ell}(p(\tau),q(\tau)) \frac{\partial p_\ell(\tau)}{\partial p_k}(p,q)
\right).
\end{equation}
Furthermore, as $\Phi_\tau^{0,\Lambda_1}$ is a harmonic rotation and the coefficients are all bounded by a constant $C_{harm}:=\max{\{\sup{\sqrt{m_k\nu_k}},1/\inf{\sqrt{m_k\nu_k}}\}}$   (Assumption \ref{conditions} (iv) and \cite[Lemma 15]{Nuland_Ven_2023}), we may estimate
\[
\bigg{\|}\frac{\partial \tilde q_j(\tau)}{\partial p_{k}}\bigg{\|}_{\op,\infty}
\;\le\;
C_{harm}\sum_{\ell \in X} \Biggl(
\bigg{\|}\frac{\partial q_\ell(\tau)}{\partial p_{k}}\bigg{\|}_{\op,\infty}
+
\bigg{\|}\frac{\partial p_\ell(\tau)}{\partial p_{k}}\bigg{\|}_{\op,\infty}
\Biggr),
\]
and similarly for $\frac{\partial \tilde p_j(\tau)}{\partial p_{k}}(p,q)$.
The norms on the right-hand side  are estimated as follows.  If we let 
\begin{align*}
J_{\ell k}(t):=\max{\bigg{\{} \left\| \frac{\partial q_\ell(t)}{\partial q_k}\right\|_{\op,\infty}, \left\| \frac{\partial p_\ell(t)}{\partial q_k}\right\|_{\op,\infty}, \left\| \frac{\partial q_\ell(t)}{\partial p_k}\right\|_{\op,\infty}, \left\| \frac{\partial p_\ell(t)}{\partial p_k}\right\|_{\op,\infty}\bigg{\}}}
\end{align*}
it follows from the above that
$$
\bigg{\|}\frac{\partial \tilde q_j(\tau)}{\partial p_{k}}\bigg{\|}_{op}\leq 2C_{harm}\sum_{\ell\in X}J_{\ell k}(t).
$$
This entails the bound
$$
\bigg{\|}\frac{\partial \gamma_\tau^{\Lambda_1}(f)}{\partial p_{k}}\bigg{\|}_{op}\leq  4 C_{harm}\|f\|_{C^1}\sum_{\ell\in X} J_{\ell k}(t), \quad \text{for all}\quad k\in\Gamma.
$$
We conclude that
\begin{align}\label{imp88}
\bigg{\|}\left\{ \alpha_{-(t-\tau)}^{0,\Lambda_2}(\Psi(\{k,j\})), \tilde\alpha_{\tau}^{\Lambda_1}(f) \right\}\bigg{\|}_\infty\leq 4dC_{harm}\|f\|_{C^1}\|\Psi\|C_V F(d(k,j))\sum_{\ell\in X}(J_{\ell k}(t)+J_{\ell j}(t)).
\end{align}
Hence, using the decay estimates on the Jacobians from the previous sections, this further implies that, \eqref{imp88} is bounded by
\begin{align*}
 C_0'\sqrt{C_\mu} \sinh(\sqrt{C_0} |t|) F(d(k,j)) \sum_{\ell\in X} \big( F(d(\ell,k)) + F(d(\ell,j)), 
\end{align*}
where $C_0'$ absorbs the constants $4,C_V,\|f\|_{C^1},d$, $C_{harm}$ and $\|\Psi\|$. 
Continuing from the previous estimate, integrating in time yields
\begin{align}\label{almost}
&\int_0^t \bigg{\|}\left\{ \alpha_{-(t-\tau)}^{0,\Lambda_2}(\Psi(\{k,j\})), \tilde\alpha_{\tau}^{\Lambda_1}(f) \right\}\bigg{\|}_\infty d\tau
\nonumber \\&\leq C_0' \sqrt{C_0} F(d(k,j))  \sum_{\ell \in X} \big( F(d(\ell,k)) + F(d(\ell,j)) \int_0^{|t|} \sinh(\sqrt{C_0} s) \, ds.
\end{align}
Using
$$
\int_0^{|t|} \sinh(\sqrt{C_0} s) \, ds \leq \frac{1}{\sqrt{C_0}}(\cos(\sqrt{C_0} |t|)-1)\leq  \frac{1}{\sqrt{C_0}}\cos(\sqrt{C_0} |t|)
$$
we get the following upper bound for \eqref{almost}, i.e.
\begin{align*}
&\leq C_0' F(d(k,j)) \sum_{\ell\in X} \big( F(d(\ell,k)) + F(d(\ell,j)) \big)  \cosh(\sqrt{C_0}|t|).
\end{align*}
Summing over all pairs $\{k,j\} \in \mathcal{B}$, where $\mathcal{B}$ is the set of interactions intersecting $\Lambda_2 \setminus \Lambda_1$, we have
\begin{align*}
&\sum_{\{k,j\} \in \mathcal{B}} \int_0^t \bigg{\|}\left\{ \alpha_{-(t-\tau)}^{0,\Lambda_2}(\Psi(\{k,j\})), \gamma_{\tau}^{\Lambda_1}(f) \right\}\bigg{\|}_\infty d\tau
\leq \\& C_0'  \cosh(\sqrt{C_0}|t|) \sum_{\{k,j\} \in \mathcal{B}} F(d(k,j)) \sum_{\ell \in X} \big( F(d(\ell,k)) + F(d(\ell,j)) \big).
\end{align*}
Using the convolution property of $F_0$, this quantity is bounded by
$$
\leq 2C_{F}C_0' \cosh(\sqrt{C_0} |t|)\sum_{\ell \in X} \sum_{x \in \Lambda_2 \setminus \Lambda_1} F(d(\ell,x)).
$$
Since $\ell \in X \subset \Lambda_1$, we use the properties of $F$ (see Section \ref{Set-up}) and finiteness of $X$ to conclude that
$$
\sum_{\ell\in X}\sum_{x \in \Lambda_2 \setminus \Lambda_1} F(d(\ell,x)) \to 0,
$$
as $\Lambda_1,\Lambda_2\nearrow\Gamma$.
Putting all together,
\begin{align*}
&\sum_{\{k,l\} \in \mathcal{B}} \int_0^t \bigg{\|}\left\{ \alpha_{-(t-\tau)}^{0,\Lambda_2}(\Psi(\{k,j\})), \gamma_{\tau}^{\Lambda_1}(f) \right\}\bigg{\|}_\infty d\tau
\leq\\& 2C_FC'   \cosh(\sqrt{C_0}|t|)\sum_{\ell\in X}\sum_{x \in \Lambda_2 \setminus \Lambda_1} F_\mu(d(\ell,x))\to 0, \quad (\Lambda_2,\Lambda_1\nearrow\Gamma).
\end{align*}

From the right-hand side one immediately sees that
 $\gamma_t^{\Lambda}$ is a Cauchy net in $\CR(\Omega)$, hence convergent.
\\\\
Finally, recall
\[
\gamma_t^\Lambda(f) := \alpha_{-t}^{0,\Lambda} \circ \alpha_t^\Lambda(f) \quad \implies \quad
\alpha_t^\Lambda(f) = \gamma_t^{\Lambda}(f\circ \Phi_t^{0,\Lambda})=\gamma_t^{\Lambda}(f\circ \Phi_t^{0,X})
\]
where $\Phi_t^{0,X}$ is the free harmonic oscillator flow restricted to $X$. Note that $f\circ \Phi_t^{0,X}\in \CR(\Omega_X)$ is a function independent of $\Lambda$.
Hence, if $\Lambda_1,\Lambda_2$ grow large,
\[
\|\alpha_t^{\Lambda_2}(f) - \alpha_t^{\Lambda_1}(f) \|_\infty\to 0,
\]
as desired.
 From our bounds it is clear that this holds for any compact time interval.
To conclude we note that the facts
\begin{itemize}
    \item  $\alpha_t(\cdot)$ extends to all of $\CR(\Omega)$;
    \item  $\alpha_t(\cdot)$ is an isometric $*$-homomorphism;
    \item $t\mapsto \alpha_t$ is a  group homomorphism,
\end{itemize}
are direct consequences of the proof of \cite[Thm. 14]{Nuland_Ven_2023}. Hence,  $t\mapsto \alpha_t$ is a one-parameter subgroup of automorphisms of $\CR(\Omega)$ and thus describes the infinite-volume dynamics. This finishes the proof.
 \end{proof}

\noindent
{\textbf{Acknowledgements.} Both authors acknowledge the  support of Jean-Bernard Brue and Teun van Nuland for their feedback. We also thank Nicolo' Drago for pointing out the paper by Marchioro et al.}

\end{document}